\documentclass[acmtog]{acmart}
\acmSubmissionID{453}

\newcommand{\maths}[1]{{\mathbb #1}}  

\newcommand{\RR}{\maths{R}}

\newcommand{\CC}{\maths{C}}

\newtheorem{Proposition}{Proposition}[section]  
\newtheorem{Def sec}{Definition}[section]   
\newtheorem{Definition}{Definition}[section]

\newcommand{\revised}[1]{{\color{black} #1}}
\newcommand{\revisedt}[1]{\textcolor{black}{#1}}

\usepackage{booktabs} 
\usepackage[english]{babel}
\usepackage[utf8]{inputenc}
\usepackage[T1]{fontenc}     
\usepackage{graphicx,wrapfig,lipsum}
\usepackage{subcaption}
\usepackage{float}
\usepackage{mathbbol}
\usepackage{amsthm}
\usepackage{amsmath} 
\usepackage{algorithm}
\usepackage[noend]{algpseudocode}
\usepackage{amssymb}
\usepackage[all]{xy}
\usepackage{bbold}
\usepackage{xfrac}
\usepackage{hyperref}
\usepackage{caption}
\usepackage{xcolor}

\usepackage{soul}

\setcopyright{rightsretained}

\acmISBN{123-4567-24-567/08/06}

\setcopyright{acmlicensed}
\acmJournal{TOG}
\acmYear{2018}\acmVolume{37}\acmNumber{6}\acmArticle{1}\acmMonth{11} \acmDOI{10.1145/3272127.3275102}
\acmArticle{236}

\begin{document}
\title{Multi-directional Geodesic Neural Networks via Equivariant Convolution}

\author{Adrien Poulenard}
\affiliation{%
  \institution{LIX, Ecole Polytechnique}
}
\email{adrien.poulenard@inria.fr}

\author{Maks Ovsjanikov}
\affiliation{%
  \institution{LIX, Ecole Polytechnique}
}
\email{maks@lix.polytechnique.fr}


\begin{abstract}
  We propose a novel approach for performing convolution of signals on curved surfaces and
  show its utility in a variety of geometric deep learning applications. Key to our construction is
  the notion of \revised{directional} functions defined on the surface, which extend the classic real-valued
  signals and which can be naturally convolved with with real-valued template functions. As a result, 
  rather than trying to fix a canonical orientation or only keeping the maximal response across all 
  alignments of a 2D template at every point of the surface, as done in previous works, we show how 
  information across all rotations can be kept across different layers of the neural network. Our construction, 
  which we call \emph{multi-directional geodesic convolution}, or \emph{directional convolution} for short, 
  allows, in particular, to propagate and relate directional information across layers and thus different regions on the shape. 
  We first define \revised{directional convolution} in the continuous setting, prove its key properties and then 
  show how it can be implemented in practice, for shapes represented as triangle meshes. We evaluate \revised{directional convolution} in a wide variety of learning  scenarios ranging from classification of signals on surfaces, to shape segmentation and shape  matching, where we show a significant improvement over several baselines.
\end{abstract}

%
%




\maketitle

\section{Introduction}
The success of convolutional neural networks (CNNs) for image processing tasks \cite{krizhevsky2012imagenet} has brought attention from the geometry processing community. In recent years multiple techniques have been developed to reproduce the success of CNNs in the context of geometry of curved surfaces with applications including shape recognition \cite{su2015multi}, segmentation \cite{kalo17,maron2017convolutional} or shape matching \cite{boscaini2015learning}, among many others. A key aspect of CNNs is that they rely on convolution operations. In the Euclidean domain the notion of convolution is well-defined whereas on non-Euclidean spaces, \revised{in general,} there is no direct analogue that satisfies all the same properties.

Different approaches have been proposed to overcome this limitation. Perhaps the simplest and most common consist in either performing convolution directly on the surrounding Euclidean 3D space (using so-called volumetric approaches \cite{maturana2015voxnet}, \cite{qi2016volumetric}), or constructing multiple projections (views) of an embedded object from different angles and applying standard CNNs in 2D \cite{su2015multi}. This idea has also been extended to using more general mappings onto canonical domains, including the plane  \cite{sinha2016deep,ezuz2017gwcnn}, or e.g. toric domains which admit a global parameterization and where convolution can be defined naturally \cite{maron2017convolutional}. Unfortunately, such mappings can induce significant distortion and might be restricted to only certain topological classes.

\begin{figure}
  \centering
  \subcaptionbox{\revised{GCNN 3D}}{\includegraphics[width=0.3\linewidth]{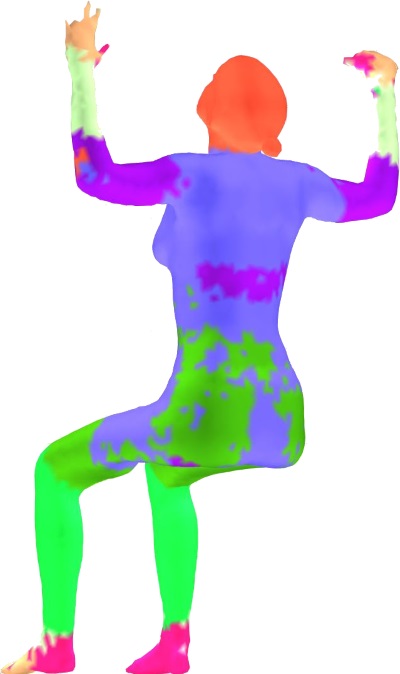}}\hfill%
  \subcaptionbox{\revised{Ours 3D}}{\includegraphics[width=0.3\linewidth]{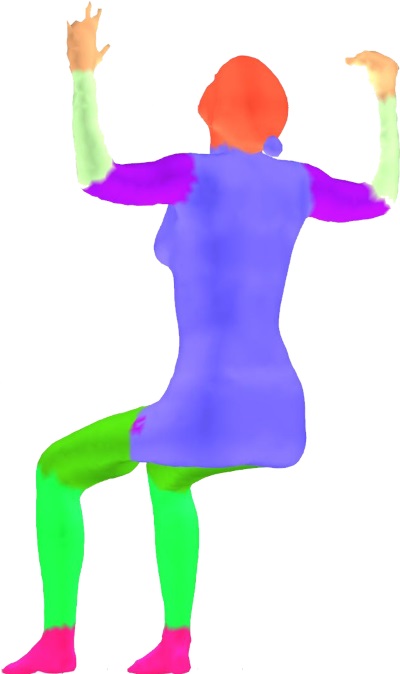}}\hfill%
  \subcaptionbox{Ground truth}{\includegraphics[width=0.3\linewidth]{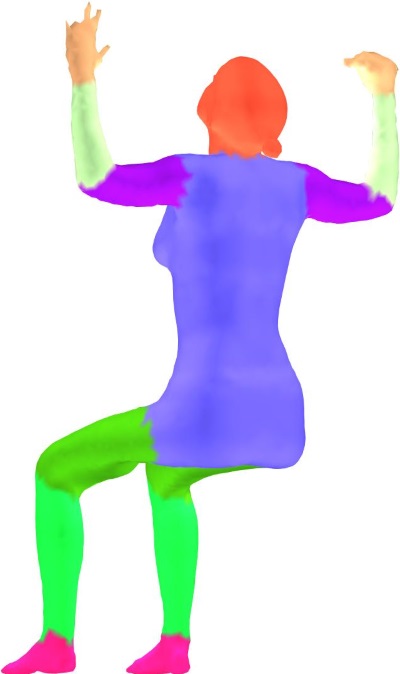}}
  \caption{Our approach (middle) enables a variety of learning applications, including shape segmentation using ``weak'' signals, such as 3D coordinates as input, via a richer and more accurate notion of convolution on surfaces. This leads to practical improvements compared to previous methods such as GCNN \cite{boscaini2015learning} as shown on the left.\label{fig:teaser}}\vspace{-3mm}
\end{figure}

On the other hand, several \emph{intrinsic} approaches have been proposed to define analogues of convolution directly on the manifold \cite{boscaini2015learning,MasBosBroVan15,MasBosBroVan16}. These techniques aim at learning local template (or kernel) functions, which can be mapped onto a neighborhood of each point on a surface and convolved with signals defined on the shape. These methods are general, can be applied regardless of the shape topology and moreover are not sensitive to the changes in shape embedding, making them attractive in non-rigid shape matching applications, for example. Unfortunately, general surfaces lack even \emph{local} canonical coordinate systems, which means that the mapping of a template onto the surface is defined only up to the choice of an orthonormal basis of the tangent plane at every point. To overcome this limitation, previous methods either only consider the maximal response over a certain number of rotations \cite{boscaini2015learning} or aim to resolve these ambiguities using principal curvature directions \cite{MasBosBroVan15,MasBosBroVan16}. Unfortunately such approaches can either lead to more instabilities or, in the case of angular max pooling, lose the relative orientation of the kernels across multiple convolutional layers.

%
%

In this paper, we propose to overcome this key limitation of intrinsic methods by \revised{aligning} the convolutional layers of the network. Our idea is to consider \emph{directional} (or, equivalently, \emph{angular}) functions defined on surfaces, and to define a notion of convolution for them which results, again, in \emph{directional functions} without loss of information. This allows us to impose specific canonical relations across the layers of a neural network, lifting the directional ambiguity to only the last layer. We then need to take the maximal response over all rotations \emph{only on the last layer}. This allows us to better capture the relative response at different points, leading to an overall improvement in training and test accuracy.


\section{Related Work}
Geometric Deep Learning is an emerging field aimed at applying machine learning techniques in the
context of geometric data analysis and processing. Below we review the techniques most closely
related to ours, and especially concentrate on various ways to define and use convolution on
geometric (3D) shapes and refer the interested reader to several recent surveys, e.g.
\cite{xu2016data,bronstein2017geometric}.

\subsection{Extrinsic and Volumetric Techniques}
Perhaps the most common approach for exploiting the
power of Convolutional Neural Networks (CNNs) on 3D shapes is to transform them into 2D images, for
example by rendering multiple views of the object. Some of the earliest variants of this idea
  include methods that represent shapes as unordered collections of views (or range images)
\cite{su2015multi,wei2016dense,kalo17} or exploit the panorama image representation
\cite{shi2015deeppano,sfikas2017exploiting} among others, as well as techniques based on Geometry Images \cite{sinha2016deep}, which represent the 3D geometry by mapping the coordinates onto the plane.

Another common set of methods considers 3D shapes as volumetric objects and defines convolution
simply in the Euclidean 3D space, e.g. \cite{wu20153d,maturana2015voxnet} among others. Due to the
potential memory complexity of this approach, several efficient extensions have been proposed, e.g.,
\cite{wang2017cnn,klokov2017escape}, and a comparison between view-based and volumetric approaches
has been presented in \cite{qi2016volumetric}.

Other recent techniques analyze 3D shapes simply as collections of points and define deep neural
network architectures on point clouds, including, most prominently PointNet \cite{qi2017pointnet},
PointNet++ \cite{qi2017pointnet++} \revised{and several extensions such as PointCNN
  \cite{li2018pointcnn} and Dynamic Graph CNN \cite{wang2018dynamic}}  for shape classification and
segmentation, and PCPNet \cite{guerrero2018pcpnet} for normal and curvature estimation, among others.

Despite their efficiency and accuracy in certain situations, these techniques rely directly on
the embedding of the shapes, and are thus very sensitive to changes in shape pose, which can limit
their use, for example in non-rigid shape matching. 

\subsection{Intrinsic and Graph-based Techniques}
To overcome this limitation, several \emph{intrinsic}
methods have been proposed, for defining and exploiting convolution directly on the surface of the
shape. This includes spectral methods, which exploit the relation between convolution and
multiplication in the spectral domain \cite{boscaini2015learning}, and which have also been applied
on general graphs \cite{defferrard2016convolutional}. A similar technique, treating shapes as
graphs, has been used for analyzing arbitrary shape collections \cite{syncspeccnn}, while
synchronizing their laplacian eigenbases with functional maps
\cite{ovsjanikov2012functional}. Another recent method \cite{ezuz2017gwcnn} consists in optimizing
an embedding of the shape onto the planar domain using intrinsic metric alignment
\cite{solomon2016entropic}. \revised{Finally, a very recent Surface Networks approach
  \cite{2017arXiv170510819K} is based on stacking layers consiting of combinations of features
  and their images by the Laplace or Dirac operator to exploit intrinsic and extrinsic information.}

More closely related to our approach are techniques based on \emph{local} shape parameterization,
which define convolution of a signal with a learned kernel on a region of the shape surface. A
seminal work in this direction was done in \cite{MasBosBroVan15} where the authors defined Geodesic
Convolutional Neural Networks (GCNNs), which locally align a given kernel with the shape surface at
each point, and perform convolution in the tangent plane. Unfortunately, the absence of canonical
coordinate systems on surfaces leads to a one-directional ambiguity in the alignment. To rectify
this, the authors of \cite{MasBosBroVan15} proposed to take the maximal response across all possible
alignments. Several later extensions of this approach have used different local patch
parameterizations, \cite{MasBosBroVan16,monti2017} and also used prinicipal curvature directions to
resolve the directional ambiguity. Unfortunately,  principal curvature directions can be highly
unstable, and not uniquely defined even on basic domains such as the sphere and the torus, which can
lead to over-fitting in the training. Finally, a recent approach has been proposed for defining
convolution via mapping onto a toric domain \cite{maron2017convolutional}, which admits a global
parameterization. 

\subsection{Contribution}
In our work, we show that the directional ambiguity that exists when mapping
a template (kernel) onto the surface, as done in \cite{MasBosBroVan15,MasBosBroVan16,monti2017}
\emph{can be maintained} across the layers of the deep neural network without relying on a
canonical direction choice or only keeping the maximal response across all directions. To achieve this, 
we first extend real-valued signals to more general \emph{directional functions}. We then show that
a directional function can be convolved with a template
to produce another directional function, and can thus be stacked in a deep neural network. As a
result the directional ambiguity is lifted up to the last layer, where it can be resolved by
taking the maximum response only once. We demonstrate through extensive experiments that this leads to
overall improvement in accuracy and robustness in a range of applications. Finally, we extend
previous approaches such as \cite{MasBosBroVan16,monti2017} by adding spatial pooling layers through
mesh simplification and exploiting residual learning (ResNet) blocks \cite{he2016deep} in the architecture.



\section{Convolution over manifolds}
Throughout our work, we assume that we are dealing with 3D shapes, represented as oriented
(manifold) 2D surfaces, embedded in 3D. For simplicity of the discussion, we also assume that the shapes are without boundary, although our practical implementation does not have this
limitation.

Throughout our discussion, the input signal is assumed to be a tuple of real-valued functions defined
on each shape in the collection. These can either represent some geometric descriptors, or even
simply the 3D coordinates of each point. In this section, we describe how the convolution operation
is applied to a given signal for a fixed template. Our approach follows the general structure
proposed in \cite{MasBosBroVan15}, but we highlight the key differences, arising from our use of
\revised{\emph{directional convolution}}. Finally, let us note that for simplicity we concentrate on oriented
two-dimensional manifolds, although most of our discussion can be adapted easily to more general settings.

Recall that in the standard two-dimensional Euclidean setting, the cross-correlation or convolution of a real-valued
function $f$ by a smooth compactly supported (template) function $k: \RR^2 \rightarrow \RR$ is defined for all
$x \in \RR^2$ by:
\[
\revised{(f * k) ( x )}
:=
\int_{\RR^2} f(t) k(t+x) dt.
\]
Convolution is a way to compare the function $f$ to a template function $k$ at every point. We can
reinterpret convolution in the following way: first, for every point $x$ we identify the tangent
space to $x$ with with a copy of \revised{$\RR^2$} whose origin has been shifted to $x$. Then, translation
by $x$ can be seen as linear isometry:
\[
\tau_{x}: T_x \RR^2 \rightarrow \RR^2 \simeq T_0 \RR^2, ~~ \tau_x(p) = p + x ~\forall~ p,
\]
where $\RR^2 \simeq T_0 \RR^2$ simply means that the tangent plane at the origin can be identified
with the whole of $\RR^2$. Moreover, the identity map ``centered at $x$'' can be seen as another isometry:
\[
\mathrm{Id}_x: T_x \RR^2 \rightarrow \RR^2, ~~ \mathrm{Id}_x(p) = p ~\forall~ p.
\]
Therefore can rewrite the convolution of $f$ by $k$ at $x$, as:
\[
\revised{(f * k) ( x )}
=
\langle \mathrm{Id}_x^*f,  \tau_{x}^*k \rangle_{L^2},
\]
where the supperscript $^{*}$ means the pull-back of the function with respect to the map, and $\langle, \rangle_{L^2}$
is the standard $L^{2}$ inner product. We can
now generalize this construction to two-dimensional manifolds by generalizing the maps $\mathrm{Id}_x^*$ and
$\tau_{x}^*$. At every point $x$ of a Riemannian manifold $X$ the exponential map:
\[
\exp^X_x : T_x X \rightarrow X
\]
generalizes the previous map $\mathrm{Id}_x: T_x \RR^2 \rightarrow \RR^2$ in the sense that
$ \exp^{\RR^2}_x = \mathrm{Id}_x.$ Following \cite{MasBosBroVan15}, we assume that the template
function $k$ is defined over a copy of $\RR^2$, denoted by $T_0\RR^2$, and generalize the map
$\tau_x$ by isometrically aligning this copy with the tangent plane $T_x X$ at $x$. \revised{That is,}
the map $\tau_x : T_x X\rightarrow T_0 \RR^2$ \revised{must be a linear isometry}. This map is uniquely-defined 
given a choice of the correspondence of the coordinate axes of $\RR^2$ with a coordinate frame \revised{of} $T_x X$. 
For oriented surfaces, this reduces to the alignment of one coordinate axis. 
If \revised{$e_1 := (1,0)$} is the first coordinate axis on $\RR^2$, this means that the inverse of the map
$\tau$ must send \revised{$e_1$} to all tangent spaces, and that $k$ can be pulled-back onto $T_x X$ for any $x$,
via $\tau_x^{*}k$.

Since unlike the Euclidean case, there is no global choice of \revised{reference direction} on a surface, an arbitrary
choice of the pre-image of \revised{$e_1$} on every tangent space $T_x X$ can lead to biased results, which
furthermore will not generalize across different (e.g., training and test) shapes. 

To resolve this ambiguity, the authors of \cite{MasBosBroVan15} consider the family of maps
$\tau_{x,v}$ parameterized by a choice of a unit vector $v$ in the tangent plane of $x$, which is
mapped to the first coordinate axis in $\RR^2$, i.e. $\tau_{x,v}(v) = \revised{e_1}$. They then define \revised{geodesic
convolution} by taking the maximum response of the signal to the template function $k$ mapped via
$\tau_{x,v}$ across all choices of $v$.
\begin{align}
\revised{(f \circledast k) ( x )}
= \max_v ~ \langle (\mathrm{exp}^X_x)^{*}f,  \tau_{x,v}^{*}k \rangle_{L^2}.
\end{align}
The two main advantages of this procedure are that: 1) it does not depend on the choice of
reference direction in the tangent planes, and 2) the output of $f \circledast k$ is again a real-valued
function, so that convolutions can be applied successively within a deep network.

Unfortunately, only keeping the maximum response also results in a loss of directional information,
which can make it more difficult to detect certain types of features in the signal. In particular,
since the maximum is applied independently at every point, directional information of the response
of the signal to the template is not shared across nearby points.

\begin{figure}[t!]
    \centering
    \includegraphics[width=0.9\linewidth]{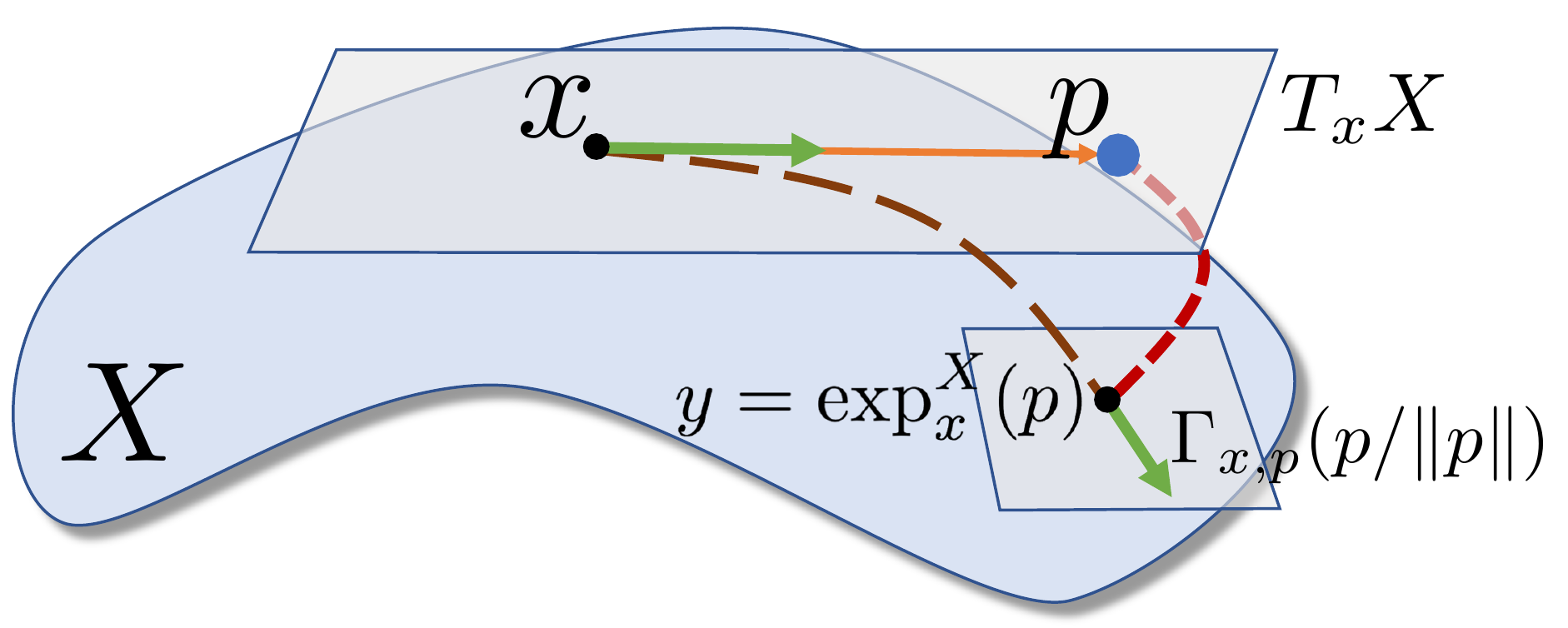}
    \vspace{-2mm}
    \caption{\revised{The \emph{completed} exponential map at a point $x \in X$ sends any non null
        tangent vector $p \in T_x X$ to the couple
        $(y = \exp_x^X(p), v = \Gamma_{x, p}(\frac{p}{||p||}))$ where $v \in T_y X$ is the result of
        parallel transport of $\frac{p}{||p||}$ along the geodesic from $x$ to $y$, with the initial
        direction given by $p$.
        \label{fig:comp_exp}}}
\end{figure}

\revised{\subsection{Convolution of Directional Functions}}
\label{subsec:dirconv}
We propose to address these limitations by considering convolution of a template function with a
more general notion of \revised{\emph{directional functions}} defined on arbitrary \revised{surfaces}.

We call a directional function, any function $\varphi(x,v)$ that depends on both the point
$x$ on a surface $X$, and on the unit direction $v$ in the tangent plane $T_x X$ at $x$. \revised{
  Clearly any real-valued function $f : X \rightarrow \RR$ can be lifted to a directional
  function $\tilde{f}$, simply by ignoring the directional argument and setting $\tilde{f}(x,v) =
  f(x)$ for any $v$. }
Our key observation is that the convolution of a \emph{directional function} $\varphi$ with respect to a
template $k : T_0 \RR^2 \rightarrow \RR$ can be defined naturally, so that the result of the
convolution is, once again, a directional function. For this we first complete the exponential map,
so that for any point $p \in T_xX$ in the tangent plane of $x$, it produces both a point $y$ on $X$
and a unit direction in the tangent plane of $y$. To achieve this we define the completed
exponential map:

\begin{align}
\overline{\exp}^X_{x}(p) = (\exp_x^X(p), ~\Gamma_{x, p}(p/\|p\|)),
\label{eq:complete_exp}
\end{align}
where $\Gamma_{x, p}(p/\|p\|)$ is the parallel transport of the unit vector $p/\|p\|$ from
the tangent plane of $x$ to the tangent plane of $y=\exp_x^X(p)$ along the geodesic between $x$ and
$y$ with initial velocity $p/\|p\|$ \revised{(see Figure \ref{fig:comp_exp})}. Thus, for any point $p$ in the tangent plane of $x$, $\overline{\exp}^X_{x}(p)$ outputs a point
$y$ on the manifold and a vector in the tangent plane of $y$. Moreover, since parallel transport
along geodesics preserves the norms of vectors, $\Gamma_{x, p}(p/\|p\|)$ must also be
a unit vector. This map is well-defined everywhere, except at the origin $p=0$. This does not pose a
problem in our setting, however, since we only use this map inside an integral.
We now define \revised{\emph{multi-directional geodesic convolution}, or \emph{directional convolution} for short,} of a template
$k : T_0 \RR^2 \rightarrow \RR$ with a directional function $\varphi$ \revised{by}:
\begin{align}
\revised{(\varphi \star k) (x,v)} = \langle (\overline{\exp}^X_{x})^{*}\varphi, \tau_{x,v}^*k \rangle_{L^2}.
\label{eq:dirconv}
\end{align}
Note that, $\psi_x = (\overline{\exp}^X_{x})^{*}\varphi$ is a real-valued function in the tangent
plane of $x$, where $\psi_x(p) = \varphi(\overline{\exp}^X_{x}(p))$ \revised{(see Figure
  \ref{fig:dir_conv_fig})}.  Moreover, note that the result of a convolution of a directional
function $\varphi$ and a template $k$ is once again a directional function, as it depends both on
the point $x$ and on the direction $v$. \revised{We also remark that directional convolution can be
  extended to regular-valued functions. Thus, for any $f:X \rightarrow \RR$ we simply set:
  $f \star k := \tilde{f} \star k$ where $\tilde{f}$ is the lifting of $f$ to a directional function,
  as described above.}

\vspace{2mm}
\subsection{\revisedt{Directional vs. Geodesic Convolution}}
As suggested in the previous section, \revised{geodesic convolution} introduced in
\cite{MasBosBroVan15} and our \revised{directional convolution} are closely related. Indeed, below
we show that \revised{directional convolution} is strictly more informative than \revised{geodesic
  convolution}. The following proposition (proved in the Appendix) shows that \revised{geodesic
  convolution} can be factorized by taking the maximal directional response of \revised{directional
  convolution} thus losing the directional information.
\begin{Proposition}
  Let $f$ a function on $X$ and $k$ a template. Denote by $\tilde{f} = (x, v) \mapsto f(x)$ the
  directional function obtained via $\tilde{f}(x,v) = f(x)$ for all unitary $v\in T_x X$. Then:
\begin{align*}
f \circledast k ( x ) = \max_{v \in T_x X} \revised{(\tilde{f} \star k)(x, v)}
\end{align*}
\label{prop:dg_to_gc}
\end{Proposition} 

\vspace{-3mm}
\revised{Intuitively, applying directional convolution allows to keep track of the direction \revised{that}
  the signal \revised{comes from}. To illustrate this, we consider the directional convolution of the
  indicator function $\mathbf{1}_{x}$ of a point $x$, by a shifted Dirac kernel $\delta_{(t, 0)}$
  for some $t > 0$. It is easy to see that the result of $\mathbf{1}_{(x, v)} \star \delta_{(t, 0)}$
  is the indicator function of the set of couples $(y, v)$ such that $y \in X$ is at (geodesic)
  distance $t$ from $x$ and $v \in T_y X$ points in the direction of $x$ (see Figure \ref{fig:dirac}
  for an illustration). Moreover, as also illustrated in Figure \ref{fig:dirac}, applying
  directional convolution by $\delta_{(t, 0)}$ multiple times will propagate the signal along
  geodesics from the source point $x$ while maintaining the directional information \revised{attached} to it.

  In contrast, after angular max-pooling directional information to the source point is lost.  This
  means that when applying the basic geodesic convolution repeatedly (e.g. when stacking multiple
  convolution layers in a neural network) the signal does not necessarily travel along geodesic
  paths.  Thus, for example,
  $(\mathbf{1}_{x} \circledast \delta_{(t, 0)}) \circledast \delta_{(t, 0)}$ is a geodesic
  \emph{ball} instead of a circle of radius $2t$ around $x$ 
(see Figure \ref{fig:dirac2}). This might
  result in a loss of information and makes stacks of filters based on geodesic convolution less
  efficient in estimating the distance between features or their relative position, compared to
  directional convolution shown in Figure \ref{fig:3c}.
}
\begin{figure}[t!]
  \centering
  \subcaptionbox{\revised{$n$ dir-conv, \\ $\mathbf{1}_x \star \delta_t \star \dots \star \delta_t$} }{\includegraphics[width=1.2in, scale=1.0]{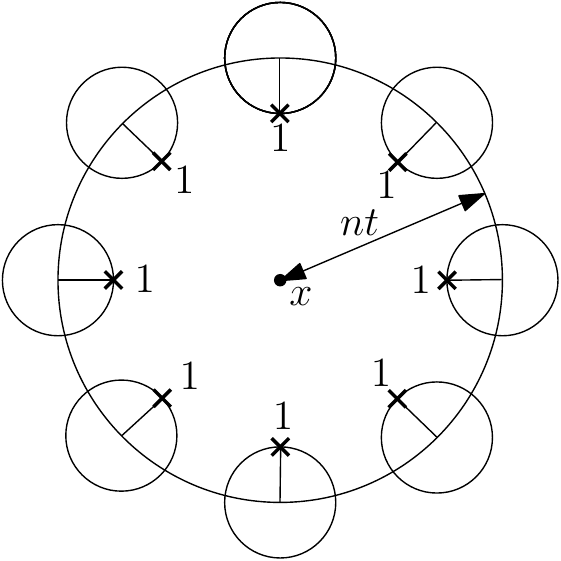}}\hfill%
  \subcaptionbox{\revised{ $n$ dir-conv, \\ 
    $(\mathbf{1}_x + \mathbf{1}_y) \star \delta_t \star \dots \star \delta_t$} \label{fig:3c}}{\includegraphics[width=1.8in,
      scale=1.0]{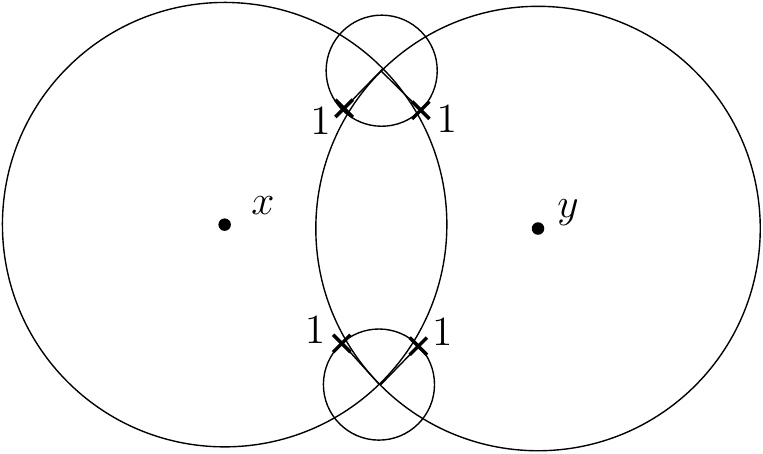}}
    \vspace{-1mm}
    \caption{\revised{(left) The result of applying $n$ times the directed convolution \revised{by the shifted Dirac kernel $\delta_{(t,
          0)}$ to the} indicator function $\mathbf{1}_x$ of point $x$. We obtain a response along the circle of radius $nt$ in the directions pointing to
        $x$. This allows to detect the presence of a feature at a given distance and to point in the
        direction of that feature. (right) The response of the convolution with two isolated signals at $x$ and $y$ is located along two circles, if they intersect the intersection points store information about the relative angle and distance between the features at $x$ and $y$.} \label{fig:dirac}}
\end{figure}
\begin{figure}[t!]
  \centering
  \subcaptionbox{$1$ conv + amp, \\ $\mathbf{1} \circledast \delta = \\ \max_v \mathbf{1} \star \delta(\bullet, v)$ \label{fig:dirac2a}}{\includegraphics[width=1.0in, scale=0.15]{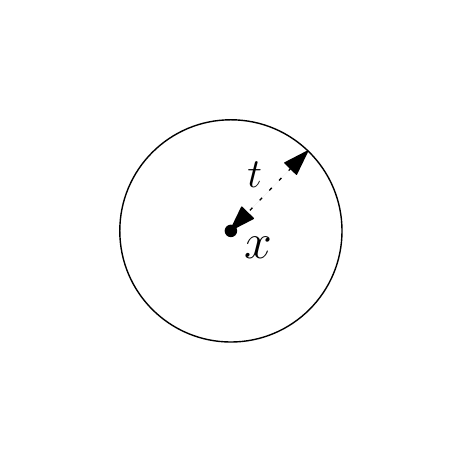}}\hfill%
  \subcaptionbox{$2$ \revised{geod}-conv, \\ $(\mathbf{1} \circledast \delta)\circledast \delta$ \label{fig:dirac2b}}{\includegraphics[width=1.0in, scale=0.15]{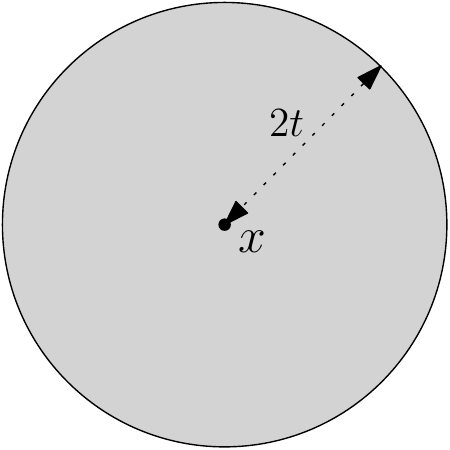}}\hfill%
    \subcaptionbox{$2$ \revised{dir}-conv + amp,\\ $\max_v (\mathbf{1} \star \delta) \star \delta (\bullet, v)$ \label{fig:dirac2c}}{\includegraphics[width=1.0in,
      scale=0.15]{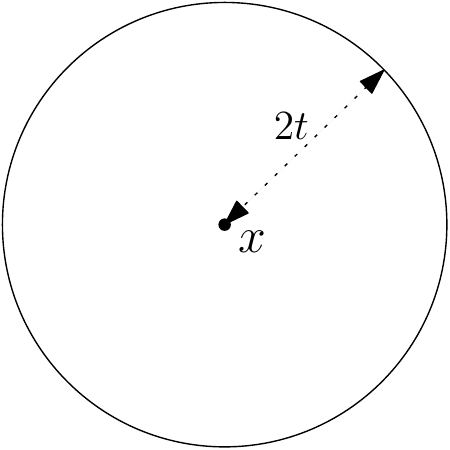}}
    \vspace{-1mm}
    \caption{\revised{Comparison between repeated geodesic (\revised{geod}, middle) and directional
        convolution (\revised{dir}, right) of the indicator of a point $\mathbf{1}$ by a shifted
        Dirac kernel $\delta = \delta_{t}$. (b) With directional convolution we obtain a response
        along the circle of radius $2t$ indicating the presence of a feature at \emph{exactly}
        distance $2t$ (c) With geodesic convolution, we obtain a response along the full disc of
        radius $2t$ only indicating the presence of a feature at distance $\leqslant 2t$.} 
 \vspace{-3mm} \label{fig:dirac2}}
\end{figure}
\revised{
  We also note briefly that directional geodesic convolution does not admit an identity kernel
  but identity can be obtained as a limit of convolutions by shifted Dirac kernels since:
\begin{align*}
\lim_{t \rightarrow 0} \overline{\exp}^X_x(tv) = (x,v).
\end{align*}
}

\begin{figure*}[t!]
  \includegraphics[width=\textwidth]{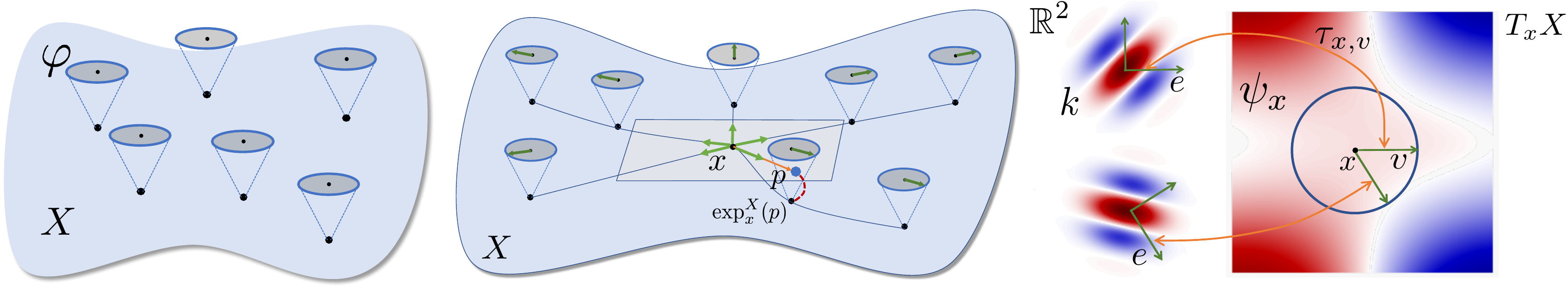}

\caption{\revised{Overview of our directional convolution: \textbf{(left)} we start with a surface
    $X$ and a directional function $\varphi$, which assigns a real value to every unit
    direction in the tangent plane of every point of $X$, shown here as a cone above some
    points. \textbf{(middle)} given a point $x \in X$, we pull back $\varphi$ onto the tangent plane
    of $x$ by computing for every point $p \in T_x X$, its image $y \in X$ via the
    exponential map, and the transport of the unit vector $p/\|p\|$ to the tangent plane of $y =
    \exp^X_x(p)$. We assign to $p$ the value of $\varphi$ at $y$ and the transported vector. By
    repeating this for every $p$ we obtain a real-valued function $\psi_x$ over the tangent plane of
    $x$. \textbf{(right)} Given a unit vector $v$ tangent at $x$ there is a unique orthogonal
    orientation-preserving linear map $\tau_{x,v}$ from $T_x X$ to $\RR^2$ sending $v$ to $e := (1,
    0)$. We pull back a kernel $k$ defined over $\RR^2$ to $T_x X$ and take its $L^2$ product with
    $\psi_x$ giving us a real-number for each unit tangent vector at $x$. Repeating this procedure
    for every $x\in X$ we obtain a new directional function which is the directional convolution of $\varphi$ by the kernel $k$.
    \label{fig:dir_conv_fig}}}
\end{figure*}

\revised{\subsection{Directional Convolution in Angular Coordinates}}

The expression in Eq. (\ref{eq:dirconv}) is coordinate-free, however to implement it we must choose a particular
coordinate system to represent directional functions. In practice, we represent tangent vectors in polar coordinates in
their tangent plane. This representation implies an arbitrary choice of reference direction $e_x \in T_x X$ in each
tangent plane. For any angle $\theta \in [0, 2\pi)$ we denote by $e_x(\theta) = R_{\theta} e_x$, the vector obtained
by rotating $e_x$ by angle $\theta$ in the tangent plane at $x$. Instead of operating with directional functions,
we then work with \emph{angular functions} as follows: given a directional function $\varphi$, and a choice of reference
direction $e_x$ for each point $x$, we define an angular function $\varphi_e$ by:
\[
\varphi_e ( x, \theta ) := \varphi(x, e_x(\theta )).
\]
\revised{that is, $\varphi_e$ is just the function $\varphi$ expressed in polar coordinates in each tangent plane with respect to the reference direction $e$.}
Note also that given \revised{reference directions $e_x \in T_x X$ for all $x \in X$}, both the exponential map and the
parallel transport can be expressed for points and vectors described in polar coordinates in each
tangent plane \revised{(see Appendix)}. Finally, given the family of reference directions \revised{$e = (e_x)_{x \in X}$} and the angular function
$\varphi_e$ we denote by $\varphi_e \star k$ its \revised{directional convolution} with kernel $k$ w.r.t. $e$. To define this
operation, we first convert $\varphi_e$ to its directional counterpart, using the reference directions $e$
and then simply apply the definition given in Eq. (\ref{eq:dirconv}) above \revised{and convert it back to an angular function using the same reference directions:
\[
\varphi_e \star k := (\varphi \star k)_e
\]
The following proposition (proved in the Appendix) shows that the result of \revised{directional convolution} of angular functions is
\emph{equivariant} with respect to the change of the reference directions. Namely:}
\begin{Proposition}
\label{prop:equivar}
Let $e_x \in T_x X$ be a family of unit tangent vectors defined at every point $x \in X$ then for any family of rotations $R_x$ of angle $\theta_x$ at point $x$ we have:
\[
\revised{(\varphi_{R e} \star k)( x, \theta ) 
= 
(\varphi \star k)_{R e}(x, \theta)
=
(\varphi_{e} \star k)(x, \theta + \theta_x )}
\]
Where $(R e)_x := R_xe_x$ at each point $x$.
\end{Proposition}

This proposition guarantees that even if we pick an arbitrary reference direction in the tangent
plane of every point $x$, the result of possibly multiple \revised{directional convolution} steps is the
same up to an angle offset. Moreover, this offset is fixed and is the angle difference between the
two reference directions. 

\revised{When directional convolution operator $\star$ is defined directly in angular coordinates,
  and angles are discretized in angular bins, as they are in our practical implementation, this
  implies that changing the reference direction at a given point by applying a circular permutation
  to the angular bins' indices will lead to the same permutation of the result of \revised{directional
    convolution}.} Thus, no further ambiguity is introduced between layers of \revised{directional convolution}, and the
  initial rotational ambiguity is lifted to the last layer. We resolve it by applying angular max
  pooling. This is crucial to learning since it allows learn the same features independently of the
  \revised{reference directions used to define the angular coordinates}.

\revised{
\subsection{Multi-Directional Convolutional Neural Networks}
Our main motivation for introducing directional convolution is to use it as a layer inside a deep
neural network for shape analysis tasks. We define a directional convolution layer simply by
replacing the standard convolution (resp. geodesic convolution) operator and regular functions by
directional functions in the convolutional layer definition. 
Crucially, since directional convolution by a template kernel results in another
directional function, we can stack \emph{directional convolutional layers} in a neural network. 

Since the input of the network is typically given as a regular signal $f: X \rightarrow \RR^d$ we first
convert it into a directional function, by setting $\tilde{f}(x, v) := f(x)$ for all
$x \in X$ and any unit vector $v \in T_x X$.  A convolutional neural network has stacks of learnable
template kernels $k^i_{pq}: \RR^2 \rightarrow \RR$, a collection of learnable bias vectors $b^i$, and
activation functions $\xi_i$. In the simplest setting, we can stack layers sequentially so that the
output of $(i+1)^{\text{st}}$ layer of the network is defined by:
\[
\left\{
    \begin{array}{ll}
        y_0(f) = \tilde{f} \\
        (y_{i+1})_p := \xi_{i+1} (\sum_q ((y_i)_q \star k^i_{pq}) + b^i_p)
    \end{array}
\right.
\]
Our goal during training is then to learn the parameters $k_{pq}^i$ and $b^i$ so that the output function
minimizes some error for a set of examples. In practice, we consider more complex
architectures, as described in Section \ref{sec:eval}. In all cases, angular max pooling is applied
to the last layer $y_n$ to obtain a regular signal over $X$. We call networks based on
directional convolution MDGCNN for Multi Directional Geodesic Convolutional Neural Networks. }

\revised{\section{\emph{\emph{directional convolution}} over Discrete Surfaces}}

We assume that all shapes are represented as connected manifold triangle meshes, possibly with
boundary. In this section we describe our general approach for implementing \revised{directional convolution}
in practice. To simplify the presentation, we describe only the main steps necessary for the
implementation, and defer the exact implementation details to Section \revised{\ref{app:impl_details}} in the Appendix.

To implement \revised{directional convolution}, we need to decide on the discrete representation of template
functions, angular functions on the surface, the exponential map and parallel transport. These are
described as follows:

\vspace{2mm}
\subsection{Template Functions}
To discretize template functions of the form $k: \RR^2 \rightarrow \RR$,
we represent them via windows of discrete polar coordinates in the plane: $(\rho_i, \theta_j)$, where
$\rho_i := \frac{(i+1)R}{N_{\rho} + 1}$ is a set of radii, for $i$ varying between $0$ and $N_{\rho} - 1$ bounded by the window radius $R$ and $\theta_j := \frac{2j\pi}{N_{\theta}}$ is a uniform discrete sampling of angles $[0, 2\pi)$. This means that each $k$ associates a real value to each pair $(\rho_i, \theta_j)$, and can therefore be stored as a matrix of size $N_\rho, N_\theta$. Note that during training these are the unknowns that need to be trained.

\subsection{Angular Functions}
Angular functions are real-valued functions that depend on the point on a
given surface and on a direction in its tangent plane. For a mesh with $N_v$ vertices, we represent
these as matrices of size $(N_v, N_{\theta})$, where $N_\theta$ is defines a discrete angular sampling, in the same way as for template functions above, but with respect to some arbitrary reference direction in the tangent plane of each point. 

\subsection{Exponential Map}
To discretize the exponential map, similarly to previous approaches, we use
Geodesic Polar Coordinates (GPC), which associate a window to every point, and represent other
points in its neighborhood, through the geodesic distance $\rho$ and angle $\theta$ to the base point
in its window. We discretize GPC using the same set of discrete polar pairs $(\rho_i, \theta_j)$ as
for the template functions. This means, in particular, that given a vertex $i$, the points in its
GPC might not fall on the vertices of the mesh. We therefore use barycentric coordinates inside
triangles to interpolate values at points in the GPC windows, based on the values at the vertices of
the mesh. This procedure is crucial since it helps us gain resilience against changes in mesh
structure.  We can model the GPC using a tensor $E$ of size $(N_v, N_\rho, N_\theta, N_v)$,
such that $E_{vijw}$ stores, the barycentric coordinates \revised{with respect to vertex $w$} of the point \revised{having polar coordinates $\rho_i$, $\theta_j$} in the GPC of the window of vertex $v$. Note that
by definition, generically, $E_{vijw}$ will have exactly three non-zero values for fixed $v,i,j$.

\subsection{Parallel Transport}
Finally, we need to discretize the parallel transport operation in order to define the
discrete analogue of the (completed) exponential map $\overline{\exp}$ defined in Eq. (\ref{eq:complete_exp}). Note that
a key feature of our definition is that parallel transport needs to be defined only for unit vectors connecting a point
$p$ in the tangent plane of $v$ to the origin, which correspond, in our discretization, to the angular polar coordinate
$\theta$ of $p$. This means that we need to compute the angle difference between $\theta$ in the window of $v$ and the
corresponding angle in the GPC of the exponential map of $p$. However due to angular discretization the transported
angle does not necessarily fall into the angles $(\theta_j)_j$ of the window at $p$. We therefore interpolate
between consecutive discrete angles to compute angular functions.  Similarly to the exponential map above, the
parallel transport can be discretized as a $5D$ tensor $\Gamma$ of size $(N_v, N_{\rho}, N_{\theta}, N_v, N_{\theta})$
where $\Gamma_{vijwk}$ stores the interpolating coordinate \revised{with respect to the
angle $\theta_k$ at vertex $w$} of the geodesic parallel transport of angle $\theta_j$ to the point with polar coordinates $(\rho_i,\theta_j)$ in the GPC of the window of vertex $v$. The  (completed) angular exponential map $\overline{\exp}$ can then be
interpolated in both the spatial and angular domains as a 5D tensor $\overline{E}$ of size
$(N_v, N_{\rho}, N_{\theta}, N_v, N_{\theta})$ such that:
\begin{align*}
\overline{E}_{vijwk} := E_{vijw}\Gamma_{vijwk}.
\end{align*}
With these constructions at hand, discretizing the generalized convolution, defined in Eq. (\ref{eq:dirconv}) can be done simply by matrix multiplication. Note that aligning the template
function with the tangent plane at $x$ via a map $\tau_{x}$ simply corresponds to aligning the angular coordinate of the template with the angular coordinate of the GPC at $x$. While
aligning a rotated template function via $\tau_{x,v}$ corresponds to applying a circular
permutation to the angular indexing of $k$. We define \revised{geodesic convolution} of a template $K$ and a function $f$ by:
\begin{align*}
\revised{(f \circledast k)_{v}} := \max_{k} \sum_{ij} f_w E_{vijk} K_{i,(j+k) \ \mathrm{mod} \ N_{\theta}}
\end{align*}
Likewise \revised{directional convolution} of an angular function $\varphi$ by $K$ is defined \revised{by}:
\begin{align*}
\revised{(\varphi \star K)_{vl}} = \sum_{ij} \varphi_{wk}\overline{E}_{vijwk}.K_{i, (j+l) \ \mathrm{mod} \ N_{\theta}}.
\end{align*}
We also note that the above definitions extend to the case of multidimensional filters and matrix-valued templates. Remark that only the triangles supporting the window points will contribute to the result of the convolution, while the value of the signal at the central point will not be taken into account. To achieve this, we add the result of a dense layer to the convolution. That is we multiply the signal at the central point (and direction for \revised{dir-conv}) by the some fixed (learned) matrix and add the resulting vector to the result of the convolution at every point.

\vspace{3mm}
\subsection{Spatial Pooling}
In traditional CNNs pooling layers refer to a way of sampling the signal. They reduce the resolution
of the image by mapping groups of pixels of the original image to a single pixel of a reduced
image. The advantage is two-fold: first it allows to summarize the information over a group of
pixels and to achieve robustness to local perturbations of the signal and second it reduces
computation time and space. A common option called max-pooling is to separate the original image into
consecutive square blocks of pixels, where each block is mapped to a single pixel of the reduced
image by taking the maximal response over the block in each channel. A
closely related notion, which is easily generalizable to meshes is the notion of strided
convolution (\revised{Figure} \ref{strided conv}). It consists of spacing the pixels of the filter window by $D-1$ pixels and applying it
every $D$ pixels thus reducing the output resolution by a factor $D$.
\begin{figure}[H]
    \centering
    \includegraphics[scale=0.3]{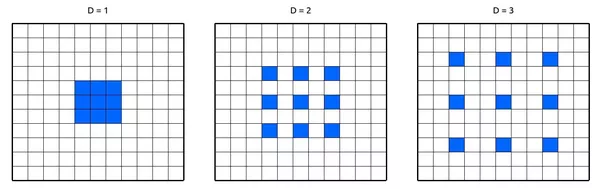}
    \vspace{-0.5mm}
    \caption{Strided convolution on images with convolution kernel in blue. Stride of $0$ (left), $1$ (middle) and $2$ (right).\vspace{-2mm}
    \label{strided conv}}
\end{figure}
We can see strided convolution as a regular convolution applied to the sub-sampled image. For meshes
we define a similar notion by transferring the signal to a coarser mesh and then applying geodesic
convolution on the new mesh. In practice, we simplify the original mesh using the classic quadratic
edge collapse approach \cite{garland97}, which also produces a mapping between the original and the
simplified meshes. We use this to transfer both a signal from the original to the simplified mesh (pooling), or from the simplified mesh to the original (un-pooling) by simply picking the value stored at the closest vertex, among those collapsed to it. 
%
To define pooling for directional (angular) functions we need to also transfer angles from the original mesh to its simplified counterpart.
To transfer local angles from a mesh $M$ to a mesh $N$ given a map from $M$ to $N$ we first compute the $3D$ rotation
sending the normal at each vertex to the normal at its image, this allows us to compare the reference directions on both
shapes in the tangent plane to $N$ and to deduce the oriented angular offset between them. We then simply transfer the
angles by adding the offset to all discretized angles and interpolating the result between consecutive discrete bins,
similarly to our construction of parallel transport.



\section{GPC and parallel transport}

To compute GPC at each vertex of the mesh we used the algorithm proposed by
\cite{melvaer2012geodesic} which is a variation of the fast marching algorithm
\cite{sethian1999level} \revised{which} allows to compute the angular coordinate as well as the geodesic
distance. We extended the algorithm of \cite{melvaer2012geodesic} to also compute the parallel
transport of angles. Fast marching-like algorithms allow to propagate information along meshes, and
rely on a local transfer subroutine. A vertex is selected among a set of candidates based on a
priority criterion, the information stored at the vertex is then propagated to some of its neighbors
based on an update criterion using the local transfer subroutine. The algorithm stops once a
certain final condition is met. In our case a vertex is updated until the radius exceeds a certain threshold $R_{\mathrm{max}} > 0$. In practice, we follow the basic approach of \cite{melvaer2012geodesic} for propagating information, but in addition to
updating the geodesic distance $\rho$ and polar angle $\theta$ we also keep track of the difference in angles between
the reference directions at the source and target points. The original algorithm \cite{melvaer2012geodesic} has a subroutine for updating the GPC
angle at a vertex inside a triangle given estimates at the two other vertices. We use the same subroutine to update the
transported angle given estimates at the two other vertices. Since the representation of directions is different at each
vertex, we must transfer estimates to every new vertex. We transfer the angles along the edges connecting the
vertices. To transport an angle from vertex $i$ to a neighbor $j$ along the edge $e_{ij}$ we first apply rotations to
the GPCs of $i$ and $j$ so that $e_{ij}$ gives the reference direction at $i$ and $e_{ji}$ gives reference direction at
$j$. The rotated GPCs at $i$ and $j$ have a relative angular offset of $\pi$ which allows to deduce the angular offset between the original GPCs. We transfer the angle at $i$ to $j$ along $e_{ij}$ by adding the angular offset between the GPCs at $i$ and $j$. We use the same edge transfer to transfer the reference direction at the source
to its neighbors in order to initialize the algorithm candidates set. This modified version allows us to compute the
angle difference between the initial and final reference directions, which in turn provides the estimate for the
parallel transport of the unit directions along geodesics between points.  

Figure \ref{fig:partransport} illustrates a single GPC window and parallel transport on a sphere. Namely, Figure \ref{fig:5a}
illustrates the angular coordinates of the GPC window via color coding, and the parallel transport of a particular
direction from the center vertex to other vertices in the window, whereas Figure \ref{fig:5b} illustrates our
discretization of the angular and radial bins.
%
\vspace{-3mm}
\begin{figure}[htp]
  \centering
  \subcaptionbox{Parallel transport of reference direction at origin and angular coordinate of the GPC. \label{fig:5a}}{\includegraphics[width=1.15in]{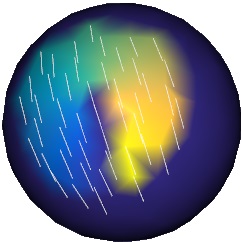}}\hfill%
  \subcaptionbox{GPC window with 2 radial and 8 angular bins and its contributing points (white) \label{fig:5b}}{\includegraphics[width=1.3in]{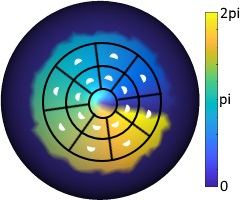}}
  \vspace{-1mm}
  \caption{An example of parallel transport and the GPC window on a sphere. The colors
    represent the angular coordinate in the GPC window.\label{fig:partransport}}
\end{figure}

\section{Evaluation}
\label{sec:eval}

\subsection{Architecture\label{subsec:arch}}

\begin{figure*}[!tbp]
  \centering
  \begin{minipage}[b]{0.3\textwidth}
    \scalebox{1.0}{\includegraphics[scale=0.35]{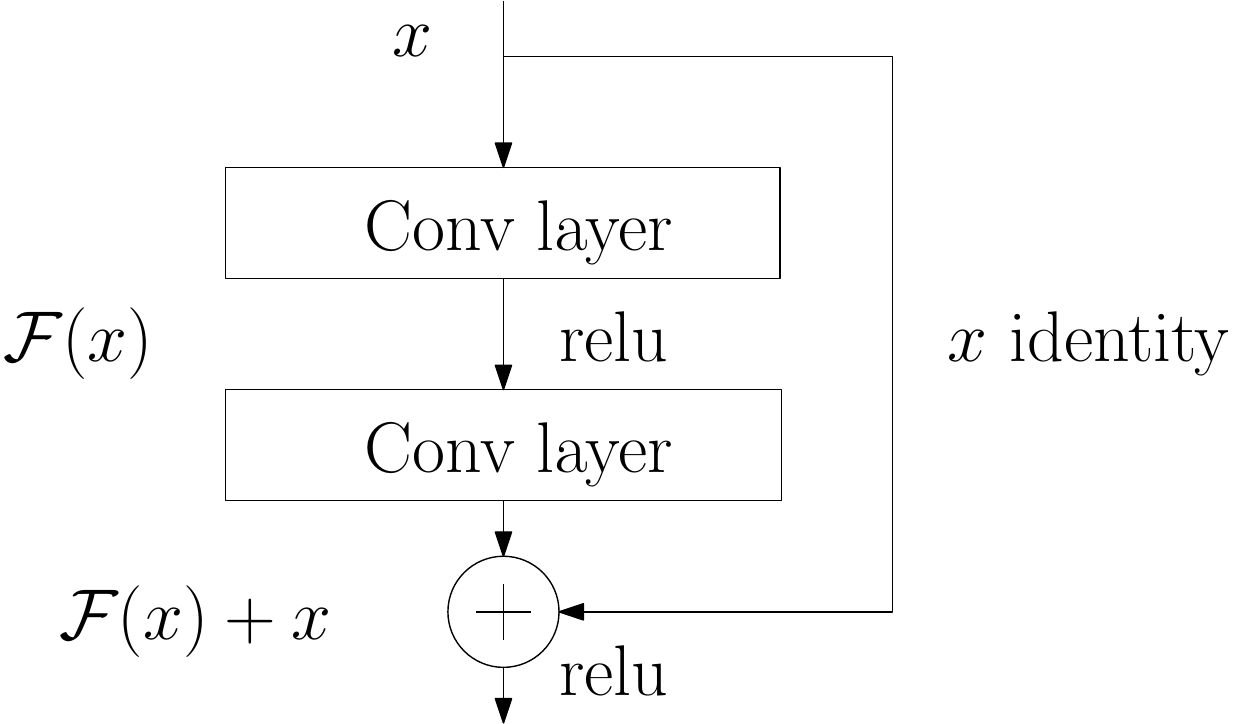}}
    \caption{ResNet block \label{ResnetBlock}}
  \end{minipage}
  \hfill
  \begin{minipage}[b]{0.3\textwidth}
   \scalebox{1.0}{\includegraphics[scale=0.30,  angle=90]{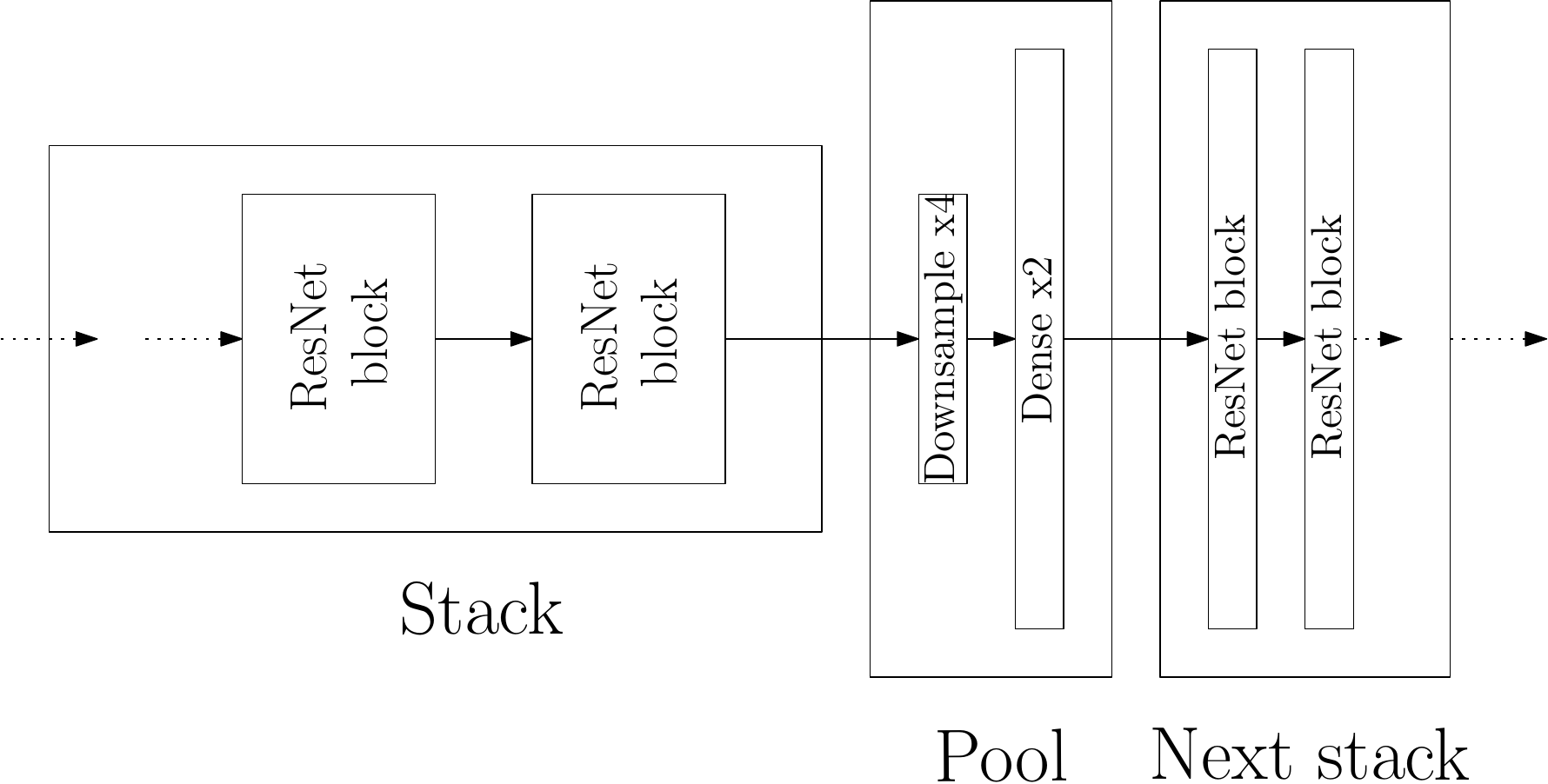}}
    \caption{Our ResNet architecture \label{ResNet}}
  \end{minipage}
    \hfill
  \begin{minipage}[b]{0.3\textwidth}
   \scalebox{1.0}{\includegraphics[scale=0.35]{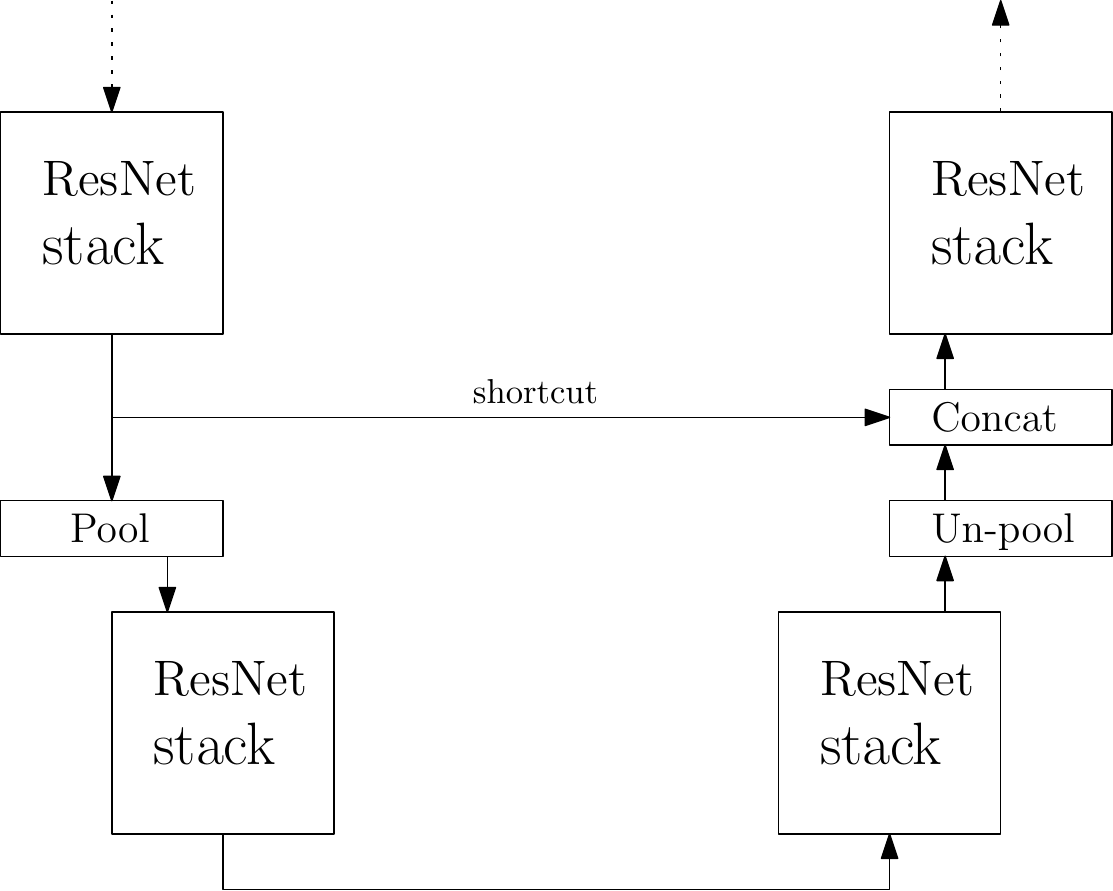}}
    \caption{U-ResNet architecture \label{Uresnet}}
  \end{minipage}
\end{figure*}

We implemented our deep learning pipeline with Keras \citep{chollet2015keras} using Tensorflow backend \citep{tensorflow2015-whitepaper}. We based our architecture on Residual networks \cite{he2016deep}. For classification tasks we used an architecture organized in stacks of ResNet blocks (Figure \ref{ResnetBlock}).
Each stack is the composition of a fixed number of ResNet blocks. After each stack the mesh is \revised{down-sampled} by a factor 4 using a pooling layer. The radius and number of filters is multiplied by two to preserve time and space complexity across stacks (Figure \ref{ResNet}).
For classification tasks we apply angular max pooling and average the signal over the shape we then apply softmax classifier on the resulting vector. For segmentation tasks we need to produce a point-wise prediction therefore the signal needs to be up-sampled back to the original shape. We combined our ResNet architecture with U-net ~\citep{ronneberger2015u}. \revised{Our U-ResNet architecture (Figure \ref{Uresnet})} consists of two blocks an encode and a decode block. The encode block is a copy of our ResNet architecture, the decode block is similar but pooling is replaced by un-pooling. After each stack the window radius is divided by 2, the signal dimension is divided by two and its up-sampled by 4. Shortcut connection are added to help keeping spatial localization information:
\subsection{Experiments}
\revised{In our experiments we compared MDGCNN to GCNN \cite{MasBosBroVan15} and PointCNN
  \cite{li2018pointcnn} in image classification and to GCNN, PointNet++
  \cite{qi2017pointnet++} and Dynamic Graph CNN \cite{wang2017cnn} in shape segmentation and
  shape matching tasks using various input features. We trained all networks using ADAM
  \cite{kingma2014adam} optimizer with learning rate 0.001. Since MDGCNN and GCNN are closely
  related we used the same architectures and window radii for both to make the comparison as fair as
  possible. } For experiments with varying domains we first center the shapes and
normalize them so that they have unit variance, then we use a fixed initial radius for the whole
dataset. The memory complexity is linear in the number of vertices, the number of radial bins and
the number of directional bins of the windows. \revised{MDGCNN} and \revised{GCNN} have similar time
and memory complexity but \revised{MDGCNN} uses more complex tensor indexing to \revised{align}
local windows and performs directional interpolation of the result of convolution. Therefore our
implementation of \revised{MDGCNN} is slightly \revised{slower} in practice. \revised{For example} in our image
classification experiment on \revised{50000 images mapped to spheres with 3000 vertexes} we observed
times of 7 min 13s for one epoch with \revised{GCNN} and 10 min 13 with \revised{MDGCNN} using \revised{a} GTX
1080 graphics card. This speed advantage however is compensated by a the significantly faster convergence of
\revised{MDGCNN as shown in Figure \ref{fig:seg_convergence}}.

\subsection{Image classification \label{img_classification}}
In our first experiment we compare \revised{MDGCNN}, \revised{GCNN} \revised{and PointCNN} on the CIFAR-10 image classification
benchmark ~\citep{krizhevsky2009learning} on different domains. We do not try to achieve state-of-the-art
performance. Our purpose is to demonstrate that \revised{MDGCNN} is able to learn complex signals over mesh domains and is superior
to \revised{GCNN}. The CIFAR-10 dataset consists of $60000$ $32$ by $32$ RGB images in $10$ classes, (airplane, automobile, bird,
cat, deer, dog, frog, horse, ship, truck) with $6000$ images per class. There are $50000$ training images and $10000$
test images. We mapped the images to two different meshes, a regular grid and a sphere. In the case of the sphere, we
parametrized two opposite hemispheres in polar coordinates and then used elliptical mapping ~\citep{fong2015analytical}
\revised{(shown in Figure \ref{fig:elliptical_map})} to map the images to both hemispheres.

\begin{wrapfigure}{l}{3.0cm}
\vspace{-13pt}
\includegraphics[scale=0.35]{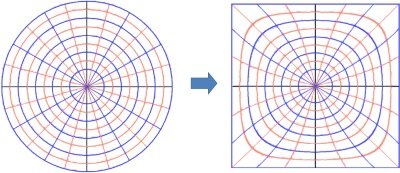}
\vspace{-20pt}
\caption{Elliptical mapping from a disc to a square \label{fig:elliptical_map}}
\vspace{-10pt}
\end{wrapfigure} Once the mapping is computed every mesh vertex is equipped with 2D coordinates which allows us to pull back the images using bilinear interpolation inside pixels. 
The resulting image on the sphere is then linearly interpolated
inside each triangle as shown on Figure \ref{cifar10 spheres}. Let us note that in both the grid and the sphere case, using
principal curvature directions to fix the reference orientation, as done in \cite{MasBosBroVan16,monti2017}, would not
be meaningful as every point is an umbilic on these domains, so that \emph{every} direction is a principal one.

\begin{figure}[H]
  \centering
  \includegraphics[width=0.55in]{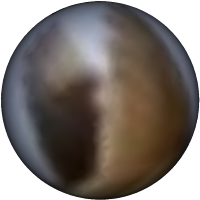}\hfill%
  \includegraphics[width=0.55in]{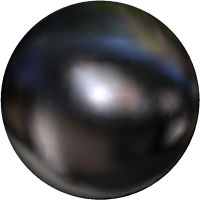}\hfill%
  \includegraphics[width=0.55in]{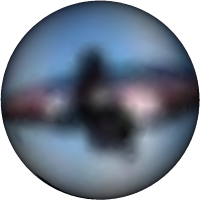}\hfill%
  \includegraphics[width=0.55in]{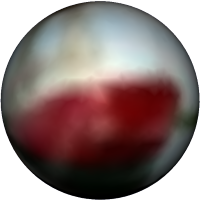}\hfill%
  \includegraphics[width=0.55in]{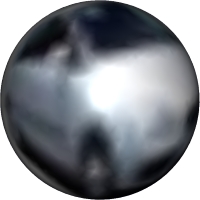}\hfill%
  \includegraphics[width=0.55in]{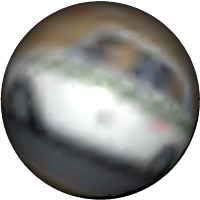}\hfill%
  \includegraphics[width=0.55in]{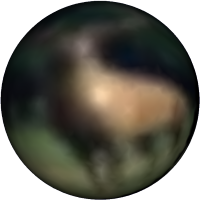}\hfill%
  \includegraphics[width=0.55in]{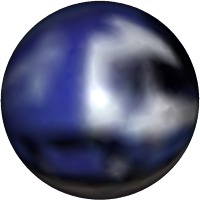}\hfill%
  \includegraphics[width=0.55in]{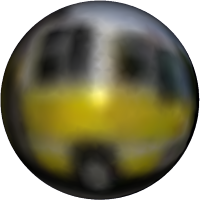}\hfill%
  \includegraphics[width=0.55in]{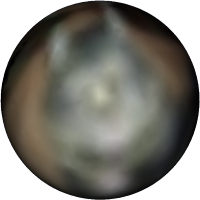}\hfill%
  \includegraphics[width=0.55in]{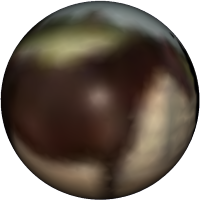}\hfill%
  \includegraphics[width=0.55in]{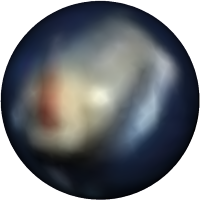}\hfill%
  \caption{Examples of CIFAR-10 images mapped to a sphere using elliptical mapping}
  \label{cifar10 spheres}
\end{figure}

We trained our ResNet architecture (Figure \ref{ResNet}) in batch of 10 images  with 3 ResNet stacks, one ResNet block (Figure \ref{ResnetBlock}) per stack and 16 filters for the first stack, we chose an initial radius of $1.8$ pixels, that is $1.8$ in the grid case and $(32 \times 1.8)/\pi$ on the sphere. The network converged after 50 epochs. 
\begin{table}
\centering
\caption{Classification accuracy on the CIFAR-10 dataset using different methods on different domains.\label{class acc}}
\begin{tabular}{ |p{1.5cm}||p{1.5cm}|p{1.5cm}|p{1.5cm}|  }
 \hline
 \multicolumn{3}{|c|}{CIFAR 10 - classification} \\
 \hline
 Method & Domain & Accuracy\\
 \hline
 {GCNN} & sphere & 0.6712\\
 \textbf{MDGCNN} & sphere & $\textbf{0.7706}$\\
 {GCNN} & grid & 0.6767\\
 \textbf{MDGCNN} & grid & $\textbf{0.7932}$\\
 \revisedt{PointCNN} & \revisedt{grid} & \revisedt{0.7669}\\
 \hline
\end{tabular}
\end{table}
The results scores in Table \ref{class acc} show a clear advantage for \revised{MDGCNN over GCNN},
we also observe similar scores \revised{for MDGCNN} on different domains despite the important
stretching introduced by the elliptical mapping of images to the sphere and the irregularity of the
sphere meshing compared to a grid mesh. \revised{In their recent work \cite{li2018pointcnn} Li et
  al. applied PointCNN to the CIFAR-10 classification benchmark. The results summarized in Table
  \ref{class acc} suggest that our method compares favorably to PointCNN in the image classification context.}
  
\subsection{Shape segmentation}
In our second experiment we compared our \revised{MDGCNN against GCNN
  and several other state-of-the-art methods: Toric cover CNN \cite{maron2017convolutional},
  PointNet++ \cite{qi2017pointnet++} and Dynamic graph CNN \cite{wang2018dynamic}.} We evaluated
all methods on the human segmentation benchmark proposed by \cite{maron2017convolutional}. This
dataset consists of 370 models from SCAPE, FAUST, MIT and Adobe Fuse \citep{fuse}. All models are
manually segmented into eight labels, three for the legs, two for the arms, one for the body and one
for the head. The test set is the 18 models from the SHREC07 dataset in human category. 
\begin{figure}[H]
  \centering
  \includegraphics[width=0.13\linewidth]{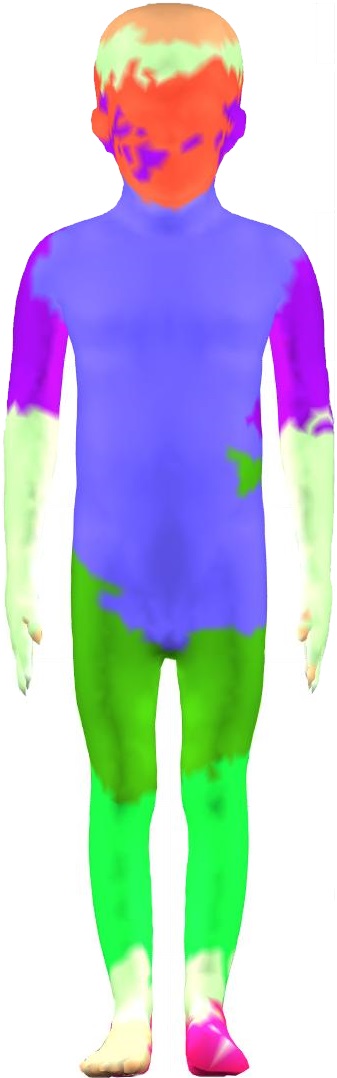}\hfill%
  \includegraphics[width=0.13\linewidth]{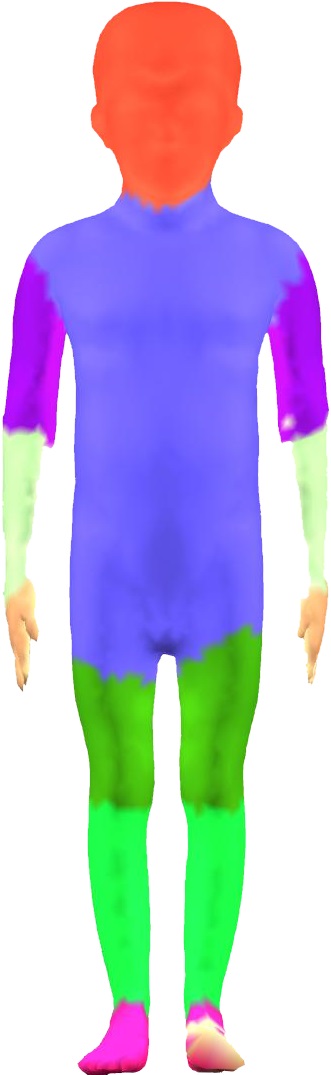}\hfill%
  \includegraphics[width=0.13\linewidth]{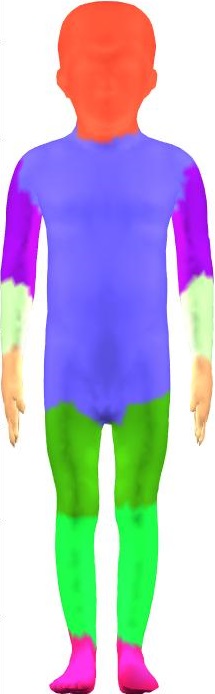}
  
  \includegraphics[width=0.13\linewidth]{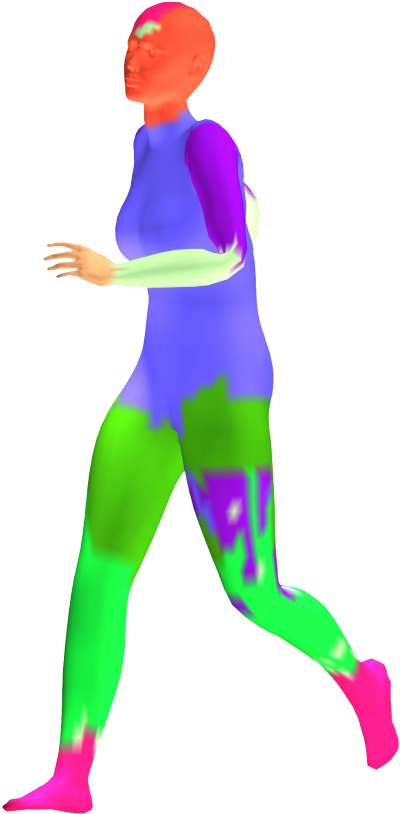}\hfill%
  \includegraphics[width=0.13\linewidth]{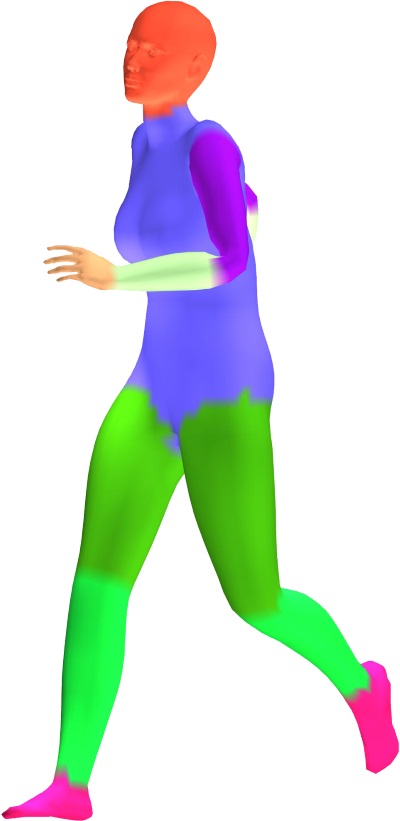}\hfill%
  \includegraphics[width=0.13\linewidth]{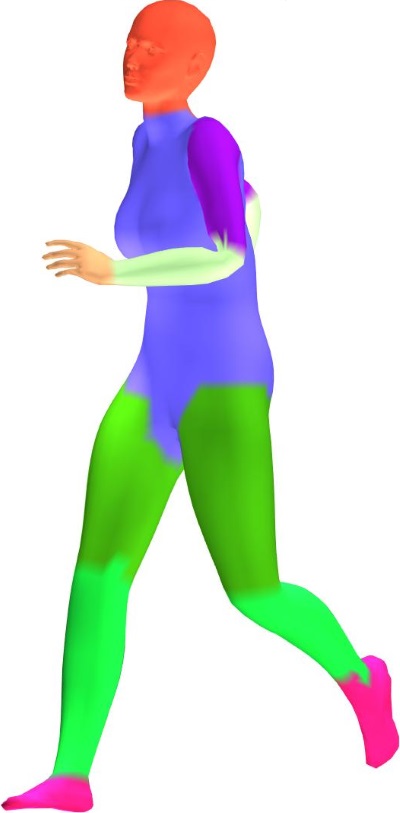}

  {\includegraphics[width=0.15\linewidth]{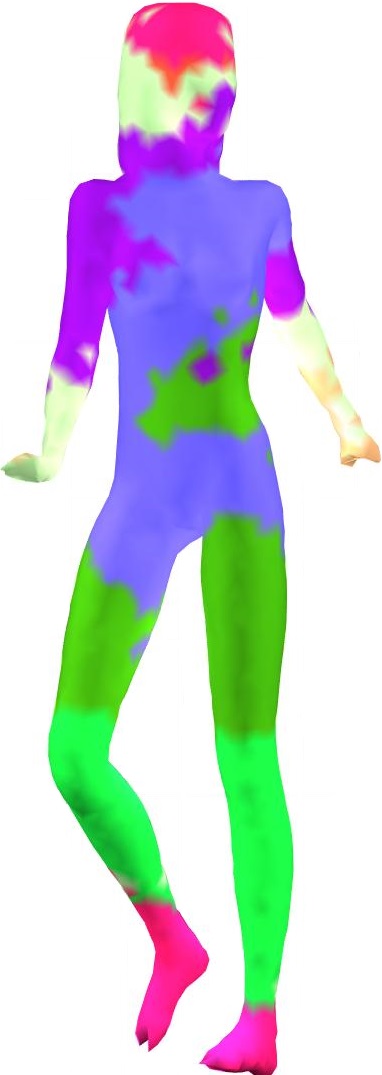}}\hfill%
  {\includegraphics[height=1.2in]{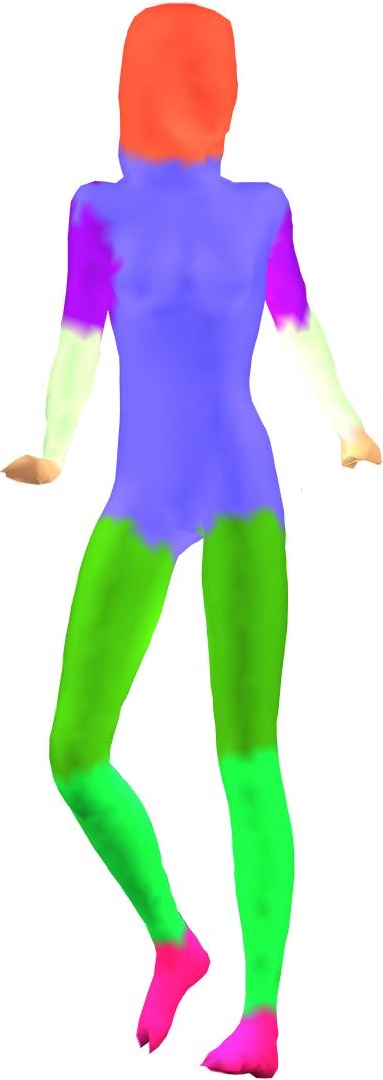}}\hfill%
  {\includegraphics[width=0.15\linewidth]{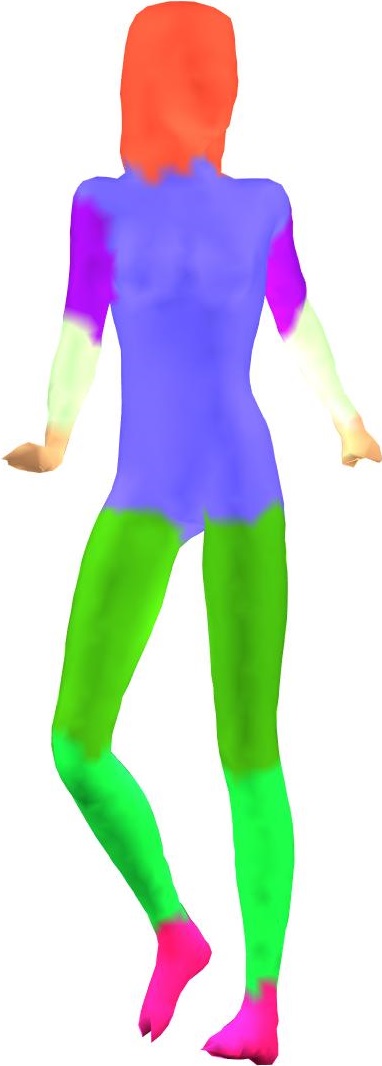}} \\

\vspace{-2mm}
  \subcaptionbox{\revised{GCNN}}{\hspace{2cm}}\hfill
 ~~~~~~~~~\subcaptionbox{\revised{MDGCNN}}{\hspace{1.6cm}}\hfill
~~~~~~~~\vspace{5mm}\subcaptionbox{Ground truth}{\hspace{2cm}}\hfill
\vspace{-3.5mm}
  \caption{Human shapes segmentation comparison between standard GCNN and \revised{MDGCNN (ours)} using the 3D coordinates of the shapes as
    input. The data is augmented by random rotations and scaling at training to ensure rigid motion invariance and
    improve robustness of the learning. \label{fig:seg_examples} \vspace{-2mm}}
\end{figure}
\begin{table}[t!]
\caption{\revised{Segmentation accuracy of several methods on the human body dataset introduced in \cite{maron2017convolutional}.}}
\begin{tabular}{ |p{1.85cm}||p{2.15cm}|p{1.cm}|p{0.9cm}|p{1.05cm}|  }
 \hline
 \multicolumn{5}{|c|}{Human body segmentation} \\
 \hline
 Method & Input feat.& \# feat. & epochs & accuracy\\
 \hline
 Toric cover  & WKS, AGD, curv.& 20+2+4 & 20 &$0.88$\\
 \revisedt{Pointnet++} & \revisedt{3D coords} & \revisedt{3} & \revisedt{200} & \revisedt{\textbf{0.9077}}  \\
 \revisedt{DynGraphCNN} & \revisedt{3D coords} & \revisedt{3} & \revisedt{200} & \revisedt{0.8972} \\
  \hline
 GCNN & 3D coords & 3 & 200 &0.7649\\
 \textbf{MDGCNN} & 3D coords & 3 & 50 &$\textbf{0.8861}$\\
 GCNN & WKS, curv. & 20+4 & 50 &0.8489\\
 \textbf{MDGCNN} & WKS, curv. & 20+4 & 50 &\textbf{0.8612}\\
 GCNN & SHOT$_{6}$  & 64 & 50 &0.3888\\
 \textbf{MDGCNN}& SHOT$_{6}$ & 64 & 50 &\textbf{0.8530}\\
 GCNN& SHOT$_{9}$ & 64 & 50 & 0.7410\\
 \textbf{MDGCNN}& SHOT$_{9}$ & 64 & 50 & \textbf{0.8879}\\
 GCNN& SHOT$_{12}$ & 64 & 50 & 0.8640\\
 \textbf{MDGCNN}& SHOT$_{12}$ & 64 & 50 &$\textbf{0.8947}$\\
 \hline
\end{tabular}
\vspace{-3mm}
\label{seg acc}
\end{table}
\begin{table}[H]
\caption{Standard deviation of test accuracy of MDGCNN and GCNN across 5 independent runs.}
\begin{tabular}{ |p{1.5cm}||p{1.5cm}|p{1.5cm}|p{1.4cm}| }
 \hline
 \multicolumn{4}{|c|}{Human body segmentation} \\
 \hline
 method & Input & epochs & standard deviation\\
 \hline
\textbf{MDGCNN} & 3D coords & 50 & \textbf{0.01059}  \\
GCNN & 3D coords & 50 & 0.11487 \\

\textbf{MDGCNN} & 3D coords & 100 & \textbf{0.02831}  \\
GCNN & 3D coords & 100 & 0.08007 \\

\textbf{MDGCNN} & 3D coords & 200 & \textbf{0.00454}  \\
GCNN & 3D coords & 200 & 0.05513 \\
 \hline
\end{tabular}
\label{seg var}
\end{table}
\begin{figure}[H]
\center
\vspace{-3mm}
  \includegraphics[width=0.9\linewidth]{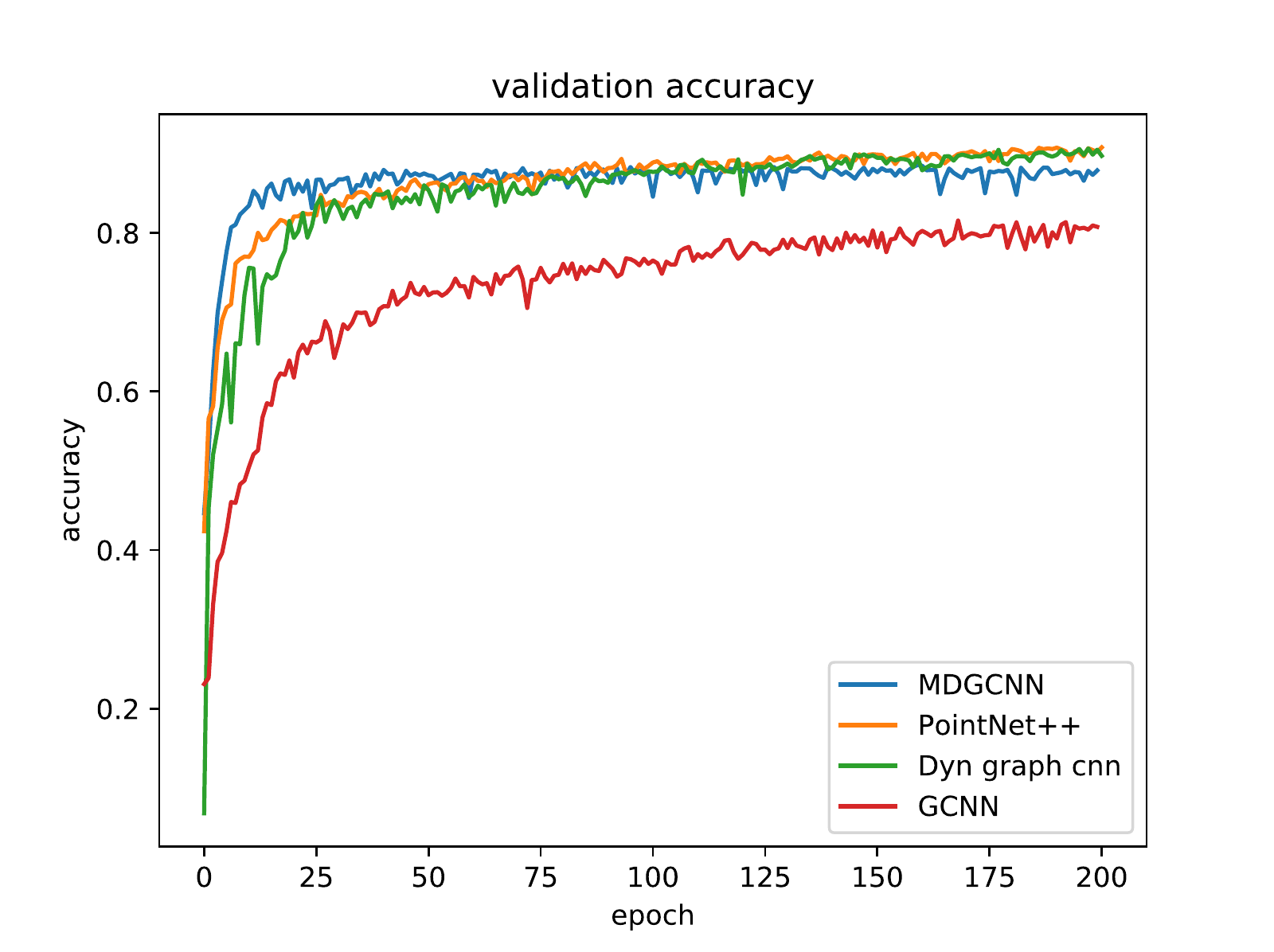}\hfill%
  \vspace{-1mm}
  \caption{\revised{Validation accuracy per epoch on the human body segmentation benchmark introduced in \cite{maron2017convolutional} using global 3D coordinates as input.}}
  \label{fig:seg_convergence}
\end{figure}
We used our U-ResNet architecture with $2\times2$ of two blocks, $16$ filters on the first layer and an initial
radius of $0.1$. Since we did not have access to the same features as used in
\cite{maron2017convolutional}, we used different inputs, SHOT \cite{salti2014shot} and WKS
\cite{aubry2011wave} descriptors as well as the $3$D coordinates of the shape, we denote by
SHOT$_{k}$ the $64$-dimensional \revised{SHOT} descriptor with window radius equal to $k$ percent of the shape
area.  For the experiments taking the $3$D coordinates of the shape as input we first apply a random
rotation and scaling between $0.85$ and $1.15$ to learn features that are robust to global
transformations. \revised{Table \ref{seg acc} summarizes the scores obtained by different methods.}

\revised{In addition to improving accuracy, we have also observed that training MDGCNN can be
  significantly more stable compared to GCNN. In Table \ref{seg var} we report the standard
  deviation of the validation accuracy of MDGCNN and GCNN across 5 independent runs with
  different number of epochs. Here the training data is fixed and the variance is only due to the
  stochastic nature of the optimization procedure. }
We observe in Figure \ref{fig:seg_examples} that \revised{MDGCNN} is  better than \revised{GCNN} at learning features that are invariant to rigid motion directly from the
3D coordinates of the shape without \emph{a priori} descriptors assumed in the data, simply via data
augmentation. Learning global features from the $3$D coordinates of the shape vertices requires aggregating information
between possibly distant points. Since \revised{directional convolution} allows better communication between distant points, \revised{MDGCNN}
noticeably outperforms \revised{GCNN} when using $3$D coordinates as input, producing much smoother results. This is also
illustrated in the qualitative results shown in Figure \ref{fig:seg_examples}. On the other hand, shape descriptors
often carry more global information about the points for example SHOT relies on local histograms counting the mesh
vertexes and normals into bins, while WKS relies on diffusion processes along the surface.  We observe close performance
in favor of \revised{MDGCNN} when using WKS. For SHOT we notice that \revised{GCNN} performance considerably degrades when reducing the
radius of the SHOT windows while \revised{MDGCNN} is able to maintain its performance much better. \revised{ PointNet++ and Dynamic Graph CNN behaved similarly on this benchmark we observed slightly slower convergence compared to MDGCNN but slightly better final results (see Figure \ref{fig:seg_convergence}).
We observe similar accuracy between MDGCNN and the Toric Cover method of
\cite{maron2017convolutional}. We note, however, that the Toric Cover method is based on
non-canonical mappings from the torus to the surface and requires considerable data augmentation by
examining many such mappings. This results in long training and prediction computation times. According to \cite{maron2017convolutional}
it takes about 5 hours for the Toric Cover method to complete one epoch at training using 6 Nvidia K80 GPUs. Using 50 different mappings, it takes 45 minutes to calculate predictions on the human class of SHREC07 while it takes 1 min 10s to train our MDGCNN network on one epoch using a single Nvidia TITAN Xp card and 2.2317s to calculate predictions on the test set.}

\subsection{Shape matching}

\begin{figure}[t!]
  \centering
\includegraphics[width=0.99\linewidth]{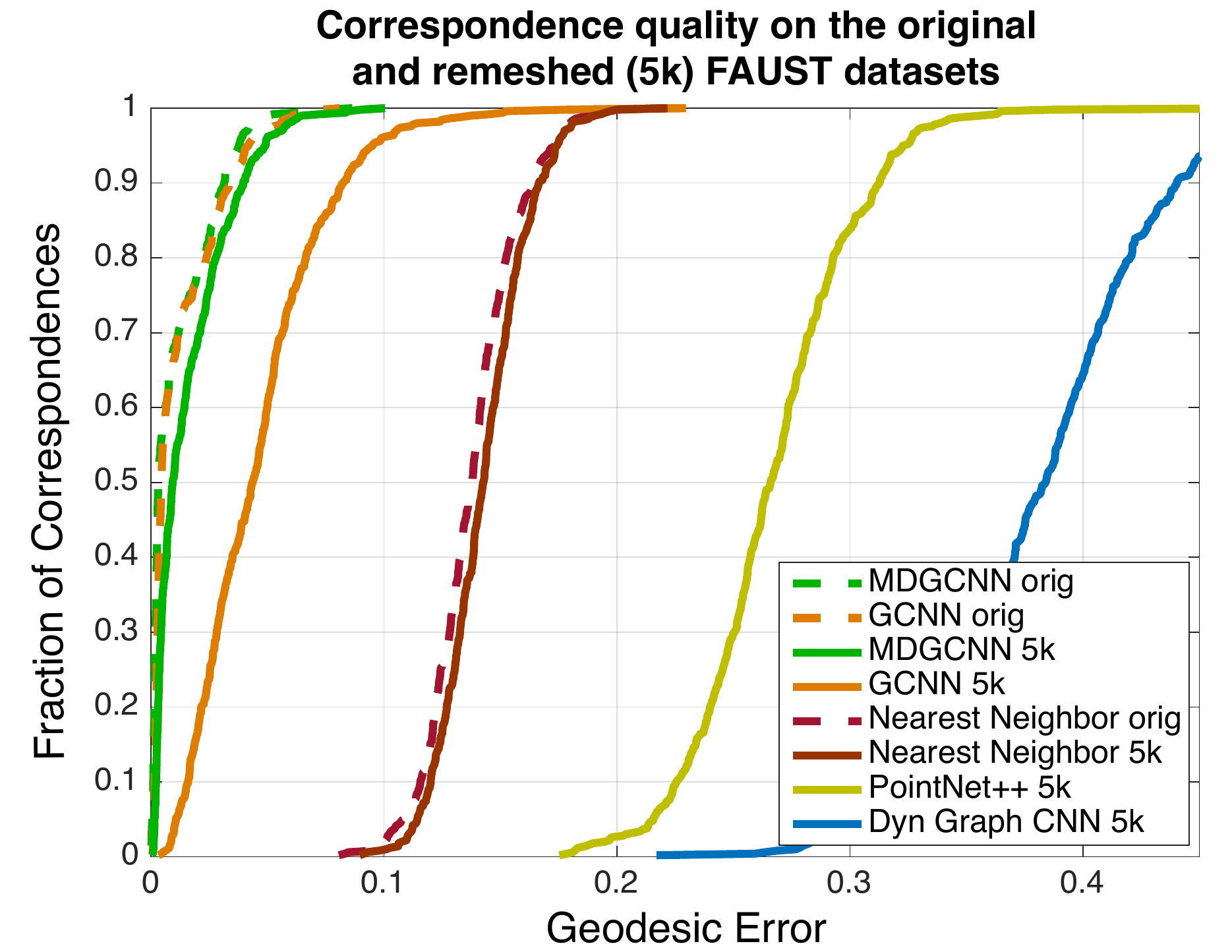}
\caption{Performance of shape correspondence on the
FAUST dataset and its re-meshed version (5k) evaluated by plotting the fraction of correspondences within a geodesic
radius of the ground truth. Higher curve corresponds to better performance.\label{fig:geod_err}}
\end{figure}

\begin{figure}[t!]
  \centering
  \includegraphics[width=0.32\linewidth]{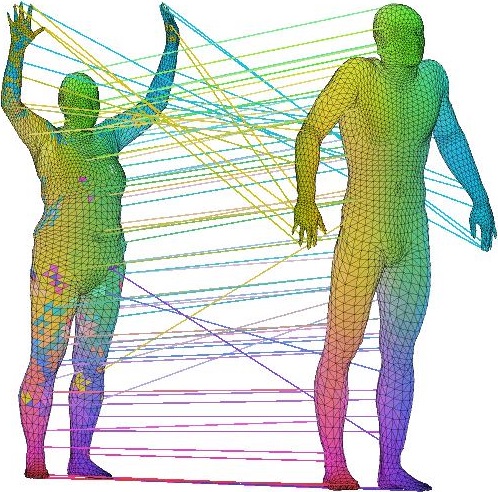}
  \includegraphics[width=0.32\linewidth]{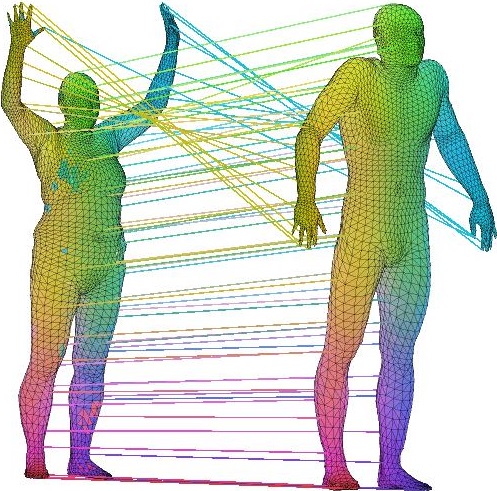}
  \includegraphics[width=0.32\linewidth]{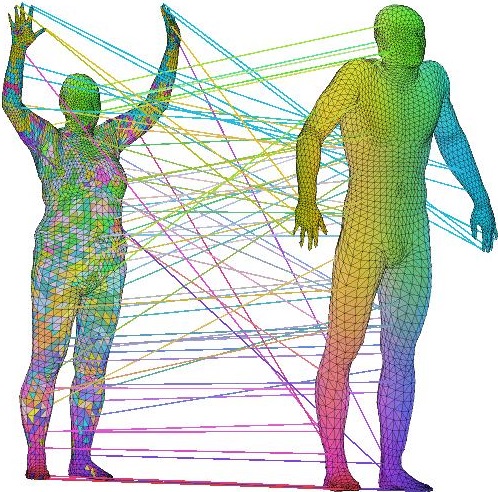}\\

    \includegraphics[width=0.32\linewidth]{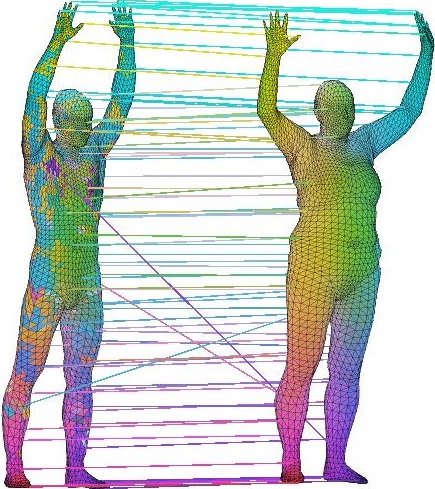}
    \includegraphics[width=0.32\linewidth]{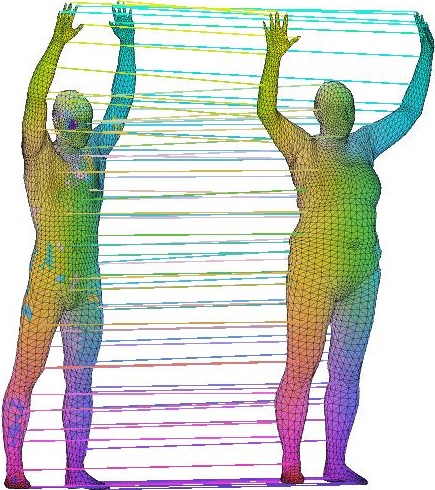}
    \includegraphics[width=0.32\linewidth]{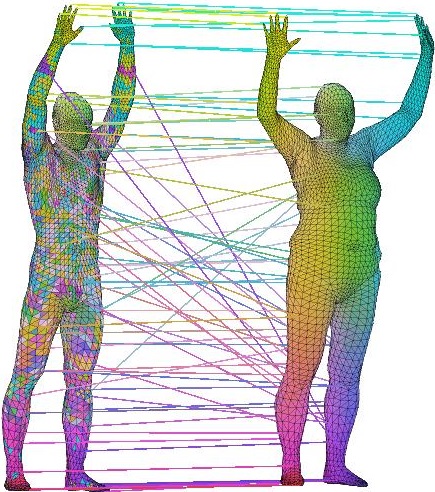}\\

  \subcaptionbox{\revised{GCNN}}{\hspace{2cm}}\hfill
  \subcaptionbox{\revised{Ours}}{\hspace{1cm}}
  \subcaptionbox{NN in descriptor space}{\hspace{4cm}}
  \vspace{-1mm}
\caption{Shape correspondence  on the
a remeshed (5k) version of the FAUST dataset using : (a)  GCNN  \cite{MasBosBroVan15} (b) Our method
\revised{MDGCNN}, and (c) using the nearest neighbor in the SHOT descriptors space. \label{fig:matching_examples}}
\end{figure} 

Finally, we also applied our pipeline in the context of non-rigid shape matching on the FAUST
dataset, used in \cite{MasBosBroVan15}. In this experiment, goal is to predict the index corresponding to each vertex in
the $0$-th shape of the dataset. The original experiment in \cite{MasBosBroVan15} used the GCNN architecture using SHOT
descriptors as input. In order to remove the bias present in the data, due to all meshes sharing the same connectivity,
we also re-meshed the FAUST shapes from $6890$ vertexes to $5000$ vertexes to evaluate the robustness of both
algorithms. We used our U-ResNet architecture with $2\times2$ stacks of two blocks \revised{with} $16$ filters on the first
layer and an initial radius of $0.1$ taking SHOT$_{12}$ as input. The network was trained for 100 epochs for \revised{MDGCNN} and
200 for \revised{GCNN}. We measured the geodesic error between the predicted labels and the ground truth on the $0$-th shape for
both the original FAUST shapes and the re-meshed ones (5k) \revised{(Figure \ref{fig:geod_err})}.

Contrary to \cite{MasBosBroVan15} we do not use post
processing on the predictions of the algorithm, and measure the accuracy directly on the output of the networks. 
As a baseline we also measured the geodesic error of correspondences obtained on $100$ random pairs
of the test set using nearest neighbors in the space of SHOT descriptors. We see that learning
techniques vastly outperform the baseline. On the original set with shapes having the same
connectivity \revised{MDGCNN} and \revised{GCNN} behave similarly. On the re-meshed set
\revised{GCNN} noticeably degrades while \revised{MDGCNN} is able to maintain its precision. Figure
\ref{fig:matching_examples} shows several examples \revised{of correspondences between pairs of
  shapes computed using different methods}. 

\revised{Figure \ref{fig:geod_err} also shows a
  comparison with PointNet++ \cite{qi2017pointnet++} and Dynamic graph CNN \cite{wang2018dynamic}
  using 3D coordinates as input (trained for 200 epochs). Note that these methods are designed for
  segmentation tasks and are not suited for shape matching in their current form.}

\subsection{Limitations \& Future Work}
Our pipeline still has important limitations. The learning process depends on the construction of local coordinate
systems which might not be suited to describe certain types of patterns possibly introducing a bottleneck to the
learning. More specifically constructions such as geodesic polar coordinates and parallel transport are purely intrinsic
based on the metric of the surface, therefore some areas that are close in the embedding space might be considered far
in this representation. A typical limiting case of purely intrinsic pipelines such as ours is when some region of the
shape is made of multiple parts that seem to merge in a single object, they might very well fail to recognize it as as a
such. We illustrate this by applying our segmentation pipeline to the bird class of PSB dataset, shown in Figure
\ref{fig:bird_bad}. From a purely \revised{topological} perspective the bird's wings are locally disconnected as they are made of
many feathers.

\begin{figure}[t!]
  \centering
  \subcaptionbox{\revised{MDGCNN} 3D}{\includegraphics[width=1.5in,height=1.5in]{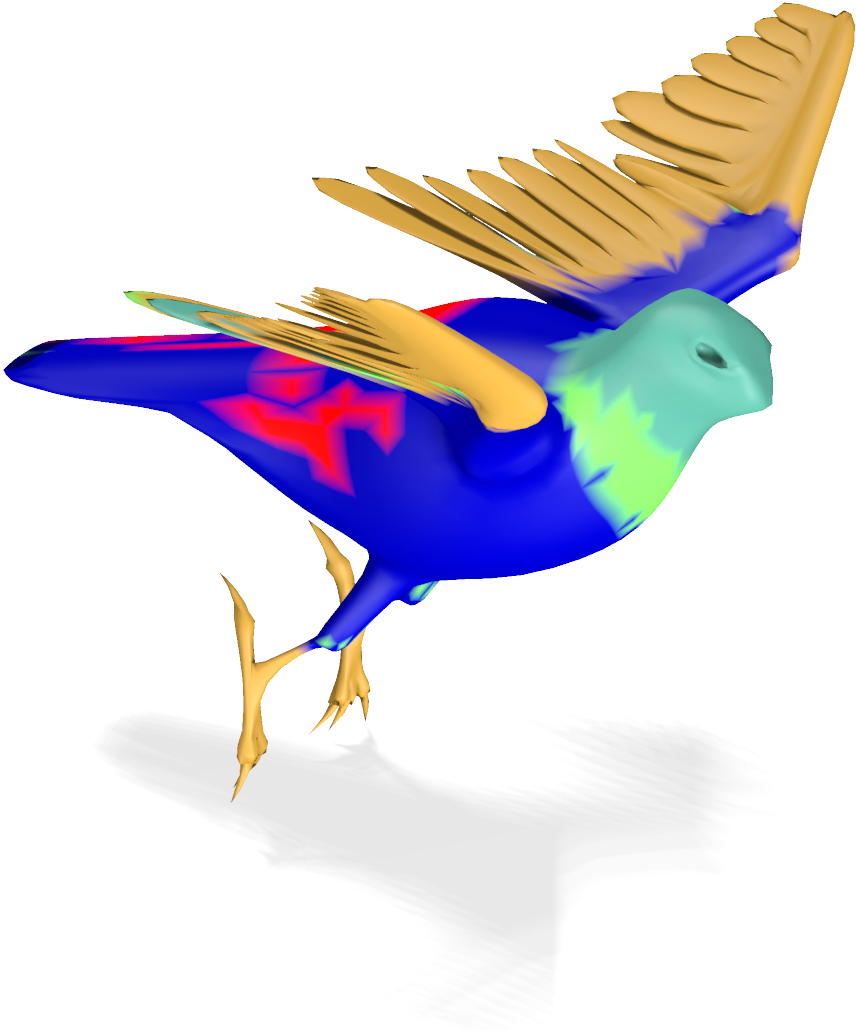}}\hfill%
  \subcaptionbox{Ground truth}{\includegraphics[width=1.5in,height=1.5in]{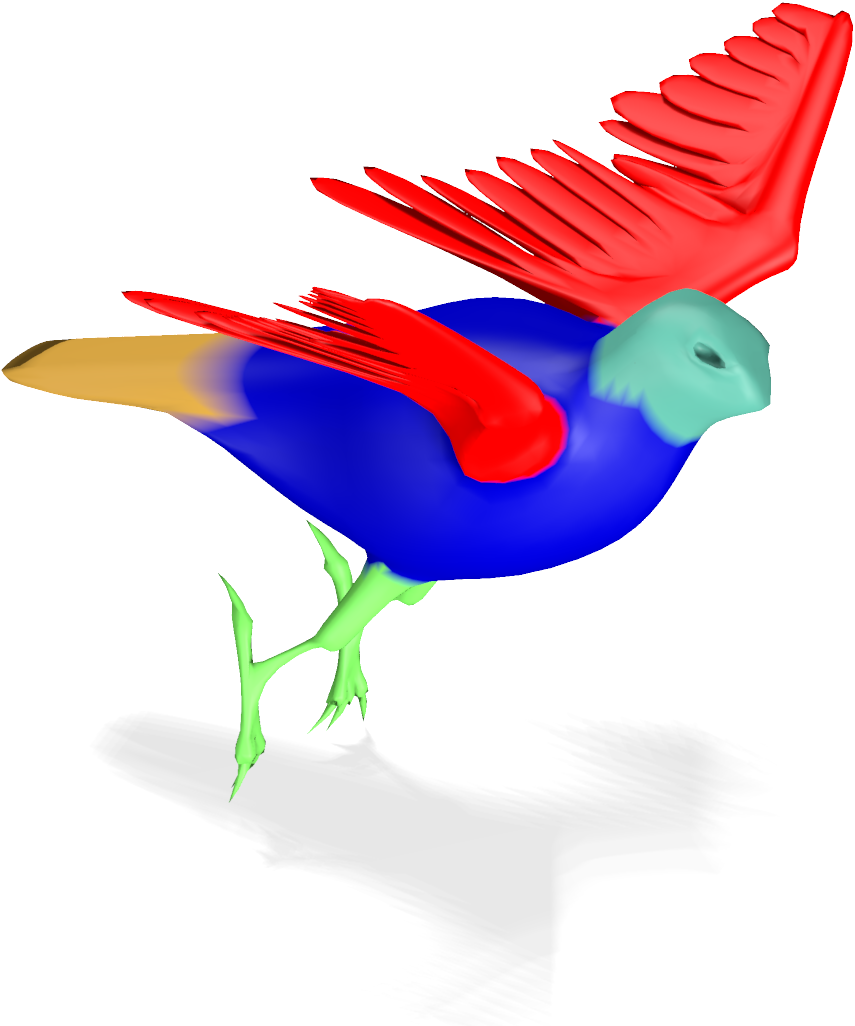}}
 \caption{Limiting case of purely intrinsic pipelines\label{fig:bird_bad}}
\end{figure}

As we can see in Figure \ref{fig:bird_bad} our pipeline most likely recognized each independent feather as a bird tail,
failing to recognize the wing in its entirety. Another limiting case of pipelines based on local coordinates is the
presence of very thin nearly degenerate parts as the bird's legs and claws. Two dimensional systems of coordinates might
not be appropriate to model \revised{near} one dimensional parts. On the other hand, Figure \ref{fig:bird_good} shows that in
the absence of such limiting cases a bird shape can be properly segmented by \revised{MDGCNN}.

\begin{figure}[t!]
  \centering
  \subcaptionbox{\revised{MDGCNN} 3D}{\includegraphics[width=1.3in, height=1.5in]{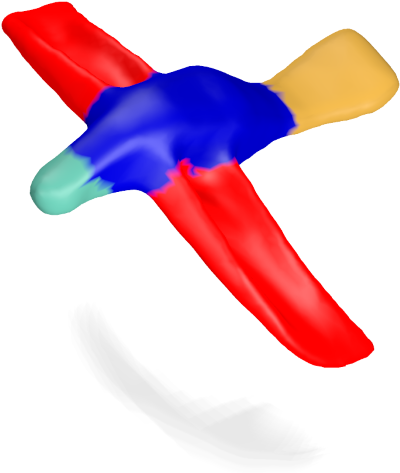}}\hfill%
  \subcaptionbox{Ground truth}{\includegraphics[width=1.3in, height=1.5in]{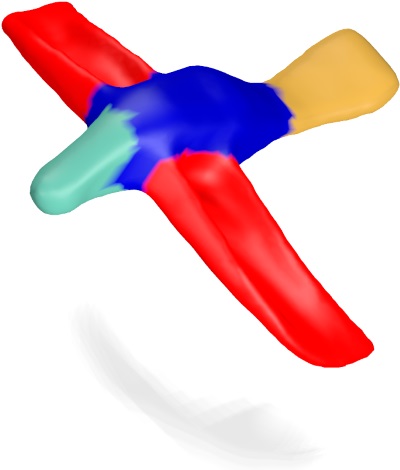}}
\caption{Successful segmentation with our approach.\label{fig:bird_good}}
\end{figure}

\revised{Another limitation shared by CNNs for image processing is the choice of scale (size) of 
  local windows. The same features can have different meaning depending on the scale at which they are detected. Since our current implementation only uses a single fixed scale, this may limit the generalization power of our method in situations where relative proportions of object parts are different from the examples in the training set. 
  In image processing the scale issue has been addressed e.g. by using the notion of
  Inception modules \cite{szegedy2015going}. We leave implementation of inception modules within the
  MDGCNN framework for future work.

  Our discretization also has some issues. For example, the parallel transport of a direction might fall between
  angular bins attached to the target vertex. For this reason we used linear interpolation between
  the two adjacent directions. However this still introduces a non-negligible error which grows with the number of layers especially for low number of angular bins. A possible alternative is to
  use the Fourier basis to represent directional signals at all points. This would allow exact
  direction transfer since rotations act linearly on the basis functions. However activation
  functions and operations such as angular max pooling would be harder to perform.}

Perhaps the most immediate, and relatively straightforward extension of our work would be to use our \revised{multi directional} approach in the context of local parameterizations via anisotropic diffusion kernels \cite{MasBosBroVan16}, but
without assuming a canonical reference direction at every point. More generally, it would be interesting to use multi-scale approaches with different ways of computing
the coordinate systems and transporting the information depending on the scale possibly using extrinsic information to help the network learning different semantic interpretations across different scales in order to improve the
overall robustness. \revised{Finally other constructions of directional convolution with potentially
  stronger properties are possible, and we give an example in the Appendix Section \ref{app:stronger_dc}.}

\section{Conclusion}
In the this work we presented a novel approach to define convolution over curved surfaces that do not admit global or canonical coordinate systems. Namely, we proposed a way to align \emph{local} systems of coordinates allowing to build and learn consistent filters that can then be naturally used across different domains. Our approach is built on the notion of directional functions, which generalize real-valued signals. We proposed a technique to convolve such directional functions with learned template (or filter) functions to produce new directional functions. This allows us to compose these convolution operations, without any loss of directional information across layers of a neural network.

We showed that our new approach compares favorably to its most direct analogues producing smoother and more robust results and can compete with more recent techniques, even when using ``weak'' input signals such as the 3D coordinates of the points. We believe that idea of \revised{multi-directional convolution can be generalized} and can open the door to addressing many other situations where the data representation is ambiguous \revised{by allowing} a neuron \revised{to} have different ways of interpreting its input \revised{and to communicate with its contributors}. For example, in our case a neuron can process its input depending on the choice of reference direction and propagate this information to its contributors. However, other types of \revised{alignment between layers} can be thought of across different types of ambiguities in the signal as well.

\begin{acks}
  \revised{ The authors would like to thank the anonymous reviewers for their valuable comments and
    helpful suggestions. Parts of this work were supported by a Google Focused Research Award and
    the ERC Starting Grant No.~758800 (EXPROTEA).}
\end{acks}

\bibliographystyle{ACM-Reference-Format}
\bibliography{bibliography}

\section{Appendix}

\revised{
\paragraph{Proof of Proposition \ref{prop:dg_to_gc}}
We have:
\[
\begin{aligned}
(\overline{\exp}^X_{x})^{*}\tilde{f}(p)
&
=
\tilde{f}(\exp^X_x(p), \Gamma_{x, p}(p / ||p||)
\\
&
=
f(\exp^X_x(p))
:=
((\exp^X_x)^*f)(p)
\end{aligned}
\]
thus:
\[
\begin{aligned}
\max_{v \in T_x X} (\tilde{f} \star k)(x, v)
&
=
\max_{v \in T_x X} \langle(\overline{\exp}^X_{x})^{*}\tilde{f}, \tau_{x,v}^*k \rangle_{L^2}.
\\
&
=
\max_{v \in T_x X} \langle(\exp^X_{x})^{*}f, \tau_{x,v}^*k \rangle_{L^2}.
\\
&
=:
f \circledast k ( x )
\end{aligned}
\]
}

\paragraph{Proof of Proposition \ref{prop:equivar}}
\revised{The first equality holds directly by definition. Namely, by applying the definition of directional convolution of the angular function $\varphi_{Re}$ with reference direction $Re$ we have:
\[
\varphi_{Re} \star k := (\varphi \star k)_{Re}.
\]
To prove the second equality observe that $Re_x$ is simply the rotation of $e_x$ by angle $\theta_x$ therefore:
\[
\begin{aligned}
(\varphi \star k)_{Re}(x, \theta)
&
:=
(\varphi \star k)(x, Re_x(\theta))
\\
&
=
(\varphi \star k)(x, e_x(\theta +\theta_x))
\\
&
=:
(\varphi \star k)_e(x, \theta +\theta_x)
\\
&
=:
(\varphi_e \star k)(x, \theta +\theta_x)
\end{aligned}
\]
which proves the proposition using the coordinate free definition of the directional convolution
operator. 

Below, we provide an alternative proof that shows that Proposition \ref{prop:equivar} still holds
when the directional convolution operator $\star$ is defined in the angular coordinate setting,
which thus provides a more direct link to the practical setting.} The
family of unit vectors $e_x \in T_x X$ defines polar coordinate systems on each tangent plane
$T_x X$. A tangent vector $p \in T_x X$ is then represented by a tuple $(r, \theta)$ where $r$ is
its radius and $\theta$ is the angle between $e_x$ and \revised{$p = re_x(\theta)$ where
  $e_x(\theta)$ denote the direct rotation of $e_x$ by angle $\theta$}. We represent unit vectors
only by their angle. Adapting the notations of Eq. (\ref{eq:dirconv}) in the polar coordinate systems
defined by $e$ we denote:
\[
\begin{aligned}
&
\tau_{x, \theta}^{e} := \tau_{x, e_x(\theta)}
\\
&
\exp_{x, e_x}^X(r, \theta) := \exp_{x}^X(r.e_x(\theta))
\end{aligned}
\]
\revised{and $\Gamma_{x, (r, \theta_1)}^{e}( \theta_2 )$ is the angle between $e_{\exp_{x, e_x}^X(r, \theta)}$ and $\Gamma_{x, re_x(\theta_1)}( e_x(\theta_2) )$ i.e.}
\[
\revised{e_{\exp_{x, e_x}^X(r, \theta)}(\Gamma_{x, (r, \theta_1)}^{e}( \theta_2 )) := \Gamma_{x, re_x(\theta_1)}( e_x(\theta_2) )}
\]
We first observe that:
\[
\begin{aligned}
&
\varphi_{R.e}(x, \theta) = \varphi_{e}(x, \theta + \theta_x),
\\
&
\tau_{x, \theta}^{R.e} = \tau_{x,\theta + \theta_x}^e,
\\
&
\exp_{x, R_x.e_x}^X(r, \theta) = \exp_{x, e_x}^X(r, \theta + \theta_x).
\\
&
\revised{\Gamma^{R.e}_{x, \theta_1}(\theta_2) = \Gamma^{e}_{x, \theta_1 + \theta_x}(\theta_2 + \theta_x )}
\end{aligned}
\]
\revised{We have:
\[
\begin{aligned}
&
(R.e)_{\exp^X_{x, R.e_x}(r, \theta_1)}(\Gamma^{R.e}_{x,(r, \theta_1)}(\theta_2))
=
\\
&
e_{\exp^X_{x, e_x}(r, \theta_1 + \theta_x)}(\Gamma^{e}_{x,(r, \theta_1 + \theta_x)}(\theta_2 + \theta_x))
=
\\
&
R.e_{\exp^X_{x,R.e_x}(r, \theta_1)}(\Gamma^{e}_{x,(r, \theta_1 + \theta_x)}(\theta_2 + \theta_x)  - \theta_{\exp^X_{x, e_x}(r, \theta + \theta_x)} )
\end{aligned}
\]
Thus:
\[
\Gamma_{x, (r, \theta_1)}^{R.e}( \theta_2 )
=
\Gamma_{x, (r, \theta_1+\theta_x)}^{e}( \theta_2 + \theta_x) - \theta_{\exp^X_{x, e_x}(r, \theta + \theta_x)}
\]}
Therefore
\[
\begin{aligned}
&
(\overline{\exp}^{X,R.e}_{x})^{*}\varphi_{R.e}(x, (r,\theta)) 
\\
&
:=
\varphi_{R.e}\big(\exp_{x, R_x.e_x}^X(r, \theta), \Gamma_{x, (r, \theta)}^{R.e}( \theta )\big)
\\
&
=
\varphi_{R.e}(\exp_{x, e_x}^X(r, \theta + \theta_x), \Gamma_{x, (r, \theta +\theta_x)}^{e}( \theta + \theta_x) - \theta_{\exp^X_{x, e_x}(r, \theta + \theta_x)})
\\
&
=
\varphi_e(\exp_{x, e_x}^X(r, \theta + \theta_x), \Gamma_{x, (r, \theta +\theta_x)}^{e}( \theta + \theta_x))
\\
&
=:
(\overline{\exp}^{X,e}_{x})^{*}\varphi_{e}(x, (r,\theta + \theta_x)). 
\end{aligned}
\]
So that:
\[
\begin{aligned}
&
\revised{(\varphi_{R.e} \star k) (x,\theta)}
\\
&
= 
\langle (\overline{\exp}^{X, R.e}_{x})^{*}\varphi_{R.e}, (\tau_{x,\theta}^{R.e})^*k \rangle.
\\
&
=
\langle (\overline{\exp}^{X, e}_{x})^{*}\varphi_{e}(\bullet, \bullet + \theta_x), (\tau_{x,\theta+\theta_x}^e)^*k \rangle
\\
&
=
\revised{(\varphi_{e} \star k) (x,\theta + \theta_x)}
\end{aligned}
\]

\subsection{Details on the implementation of convolution \label{app:impl_details}}
Our practical implementation of geodesic convolution relies on dense tensors for efficiency optimization. We describe it in the following subsection. 

Let $X$ be a triangle mesh with $N_v$ vertexes. The exponential map over $X$ is given by geodesic
polar coordinates (GPC) around each vertex. We model it as two $N_v$ by $N_v$ matrices $r$ and
$\theta$ where $r_{ij}$ and $\theta_{ij}$ represent the radius and angle at vertex $j$ of the GPC
centered at vertex $i$. The associated Euclidean coordinates $(x_{ij}, y_{ij})_{ij}$, extend inside
triangles by linear interpolation. We use the GPC to construct windows at each vertex along which we
can transfer signals on the mesh. The windows are defined by their radius $R$. The window attached to vertex $i$ consists of the points of polar coordinates are $(\frac{j.R}{N_{\rho}}, \frac{2k.\pi}{N_{\theta}})_{j \in [|1,N_{\rho}|], k \in [|1,N_{\theta}|]}$ in the GPC at $i$. Each window vertex lies inside a triangle we denote by $E_{ijkl}$ the index of the $l$-th vertex of the triangle containing $p_{ijk}$ the $jk$-th point of $i$-th window and by $W_{ijkl}$ the associated barycentric coordinate. That is:
\[
p_{ijk} = \sum_{l=1}^3 W_{ijkl} (x_{i,E_{ijkl}}, y_{i,E_{ijkl}}).
\]
A $a$-dimensional signal $f$ on the mesh consist of a $N_{v}$ by $a$ matrix, it can be pulled back to the window system by the following formula:
\[
E^*f_{ijkl} := \sum_{m=1}^3 W_{ijkm}f_{E_{ijkm},l}
\]
this can be seen as a discretization of the pull back by exponential map. Template functions are stacked in $ab$-polar kernel tensors of shape $(N_{\rho}, N_{\theta}, a, b)$ to be convolved with $a$-dimensional signals. In our context we define the \revised{geodesic convolution} of a $a$ dimensional signal $f$ by the a $ab$-polar kernel $K$ as the $b$-dimensional directional signal:
\[
(f \circledast K)_{ijk} := \sum_{r,m,l} E^*f_{irml}  K_{r,(m+j) \ \mathrm{mod} \ N_{\theta}, l,k}
\]
We adapt the original definition of geodesic convolutional layer:
\begin{Definition}[geodesic convolutional layer]
The geodesic convolutional layer of $ab$-kernel $K$, central kernel $C$ and bias vector $B$ and activation function $\xi$ transforms any $a$-dimensional signal $f$ into the $b$-dimensional signal:
\[
\revised{\mathrm{gc}_{K,C,B,\xi}(f)_{ij}} := \max_k \xi( f \circledast K_{ikj} + \sum_{l}C_{jl}f_{il} + B_j )
\]
\end{Definition}
We model the parallel transport as a $N_v$ by $N_v$ matrix $\gamma$ where $\gamma_{ij}$ is the angular offset in the angular coordinates $\theta$ of the GPC induced by the parallel transport along the geodesic joining vertex $i$ to vertex $j$. We discretize the transport of angles by the $4$D tensor $\Gamma$ of shape $(N_v, N_{\rho}, N_{\theta}, 3)$ defined by:
\[
\Gamma_{ijkl} := \frac{\gamma_{i,E_{ijkl}}}{2\pi} + \frac{k}{N_{\theta}}
\]
is the normalized angle representing the unit radial vector at $jk$-vertex of $i$-th window in its GPC. The tensor $\Gamma$ stores exact angular values however since discretized directional functions are defined using $N_{\theta}$ evenly spaced angles we consider the lower and upper transport tensors $\lfloor \Gamma \rfloor$ and $\lceil \Gamma \rceil$. We can interpolate directional signals inside the resulting angular sectors by using the fractional part $\{ \Gamma \} := \Gamma - \lfloor \Gamma \rfloor$. We define \revised{the discrete pull back} of a directional signal $\varphi$ and \revised{by the "completed" exponential map} as the tensor:
\[
\begin{aligned}
\varphi(E,\Gamma)_{ijkl} := 
\sum_{m} W_{ijkm} 
\big(
&
(1-\{\Gamma\}_{ijkm})\varphi_{E_{ijkm}, \lfloor \Gamma \rfloor_{ijkm},l}
\\
&
+
\{\Gamma\}_{ijkm}\varphi_{E_{ijkm}, \lceil \Gamma \rceil_{ijkm},l}\big)
\end{aligned}
\]
We define the discrete \revised{directional geodesic convolution} of a $a$-dimensional directional signal $\varphi$ by the $ab$-polar kernel $K$ as the $b$-dimensional directional signal:
\[
\varphi \star K_{ijk}
:=
\sum_{r,m,l} \varphi(E,\Gamma)_{irml} K_{r,(m+j) \ \mathrm{mod} \ N_{\theta}, l,k}
\]
\begin{Definition}[\revised{Directional geodesic convolutional layer}]
The \revised{directional geodesic convolution} of $ab$-polar kernel $K$, central kernel $C$ and bias vector $B \in \RR^b$ and activation function $\xi$ transforms any $a$-dimensional directional signal $\varphi$ to the $b$-dimensional directional signal:
\[
\revised{\mathrm{dir}_{K, C, B, \xi}(\varphi)_{ijk}}
:=
\xi(f \star K_{ijk} + \sum_{l}C_{kl}f_{ijl} + B_k)
\]
\end{Definition}

\subsection{\revised{A Stronger Notion of Directional Convolution \label{app:stronger_dc}}}
Several definition of \revised{directional convolution} are possible we chose to transport tangent directions along geodesic to produce a local radial vector field on the surface which is rotation invariant this simplifies implementation. There is an other natural choice \revised{to pull back directional signals to the tangent plane} we can set:
\[
\overline{\exp}^X_{x, v}(p) = (\exp_x^X(p), ~\Gamma_{x, p}(v)),
\]
and define \revised{convolution} of a directional function $\varphi$ over $X$ by a template $k$ by:
\[
\revised{(\varphi \star_X k) (x,v)} = \langle (\overline{\exp}^X_{x,v})^{*}\varphi, \tau_{x,v}^*k \rangle
\]
On the plane ($X = \RR^2$) the above formula can be simplified, we have: 
\[
\overline{\exp}^{\RR^2}_{x, v}(p) = (x+p, v)
\]
To ease notation we identify $\RR^2$ with $\CC$. The map $\tau$ simply rotates the template:
\[
\tau_{x, e^{i\theta}}^* k := k(e^{-i\theta} \bullet)
\]
therefore for any directional function $\varphi$ over $\RR^2$:
\[
\revised{(\varphi \star_{\RR^2} k) (x,e^{i\theta}) = (\varphi(\bullet, e^{i\theta}) * k(e^{-i\theta} \bullet ))(x)}
\]
Let $f$ be a function over $\RR^2$ denote by $(x,v) \mapsto \tilde{f}(x, v) := f(x)$ the associated directional function then:
\[
\revised{(\tilde{f} \star_{\RR^2} k) (x,e^{i\theta}) = (f * k)(e^{-i\theta}\bullet)(x)}
\]
The resulting directional function essentially stores the result of the usual convolution with the
rotated kernel in each direction.  Thanks to the above observation we can show that an image based
CNN using this notion of \revised{directional convolution} is equivalent to taking the maximal
response under rotations of the kernels with its standard CNN counterpart as shown by the following
proposition. For simplicity we consider one dimensional signals although the property still holds in
higher dimensions.
\begin{Proposition}
  Let $f$ a function over $\RR^2$, $(K_l)_l$ a sequence of template kernels, $(b_l)_l$ a sequence of
  real numbers $(\xi_l)_l$ a sequence of activation functions. We define the sequence
  $(\varphi_l)_l$ of directional functions over $\RR^2$ and the sequence $(f_l^{\theta})_l$ of
  function over $\RR^2$ by:
\[
\left\{
    \begin{array}{ll}
        \varphi_0 = \tilde{f}\\
        \varphi_{l+1} = \xi_{l} \circ \big(\varphi_l \star_{\RR^2} K_l + b_l\big)
    \end{array}
\right.
\left\{
    \begin{array}{ll}
        f^{\theta}_0 = f\\
        f^{\theta}_{l+1} = \xi_{l} \circ \big(\revised{(f^{\theta}_l * K_l(e^{-i\theta}\bullet)} + b_l\big)
    \end{array}
\right.
\]
Then for all $n$ we have:
\[
\varphi_n(x, e^{i\theta}) = f^{\theta}_n(x).
\]
\end{Proposition}
\begin{proof}
We proceed by recurrence. The property is true for $n = 0$ by definition, suppose it is true for $n = k$. We have:
\[
\begin{aligned}
&
\varphi_{k+1}(x, e^{i\theta}) 
\\
&
= 
\xi_k(\varphi_{k} \star_{\RR_2} K_k(x, e^{i\theta}) + b_k)
\\
&
=
\xi_k( \revised{(\varphi_k(\bullet, e^{i\theta}) * K_k)(e^{-i\theta} \bullet )(x)} + b_k)
\\
&
=
\xi_k( \revised{(f^{\theta}_k * K_k(e^{-i\theta} \bullet ))(x)} + b_k)
\\
&
:=
f^{\theta}_{k+1}(x)
\end{aligned}
\]
which proves the property for $n = k+1$ and concludes the proof.
\end{proof}


\appendix

\end{document}